\documentclass[numberwithinsect,a4paper,UKenglish]{lipics}

\usepackage{microtype}%

\title{Static Analysis for Logic-Based Dynamic Programs\footnote{The first and third author acknowledge the financial support by DFG grant SCHW 678/6-1.}}

\author{Thomas Schwentick}
\author{Nils Vortmeier}
\author{Thomas Zeume}

\affil{TU Dortmund University\\
  Germany\\
  \texttt{\{thomas.schwentick, nils.vortmeier, thomas.zeume\}@tu-dortmund.de}}

  \authorrunning{T. Schwentick, N. Vortmeier, T. Zeume} %

\Copyright{Thomas Schwentick, Nils Vortmeier and Thomas Zeume}%

\serieslogo{}%
\volumeinfo%
  {Billy Editor and Bill Editors}%
  {2}%
  {Conference title on which this volume is based on}%
  {1}%
  {1}%
  {1}%
\EventShortName{}
\DOI{10.4230/LIPIcs.xxx.yyy.p}%

\usepackage{amssymb,amsmath,xspace,enumerate,gensymb}
\usepackage{amsthm}

\usepackage[pagewise]{lineno}
\usepackage{graphicx}
\usepackage{xspace}
\usepackage{rotating}
\usepackage{listings}
\usepackage{soul}
\usepackage{ifmtarg}
\usepackage{stmaryrd} %
\usepackage{paralist}
\usepackage{chngcntr}
\usepackage{subfig}
\usepackage{tikz}
\usetikzlibrary{arrows,decorations.pathmorphing,backgrounds,calc,positioning,fit,petri,matrix,backgrounds, decorations.pathreplacing, shapes.geometric}
\usepackage{pgffor}
\usepackage{extarrows}
\usepackage{enumitem}
\usepackage{appendix}
\usepackage{hyperref}
\usepackage{algorithm}
\usepackage{algpseudocode}
\usepackage{array}

\newif\ifcomments
\newif\ifchanges
\commentsfalse\changesfalse
\commentstrue\changestrue %

\newcolumntype{C}[1]{>{\centering\let\newline\\\arraybackslash\hspace{0pt}}m{#1}}

\begin{document}
\makeatletter{}%
\makeatletter{}%
\newcommand  {\myclass} [1]  {\ensuremath{\textsc{#1}}}

\newcommand{\StaClass}[1]{\myclass{#1}\xspace}

\newcommand{\DynClass}[1]{\myclass{Dyn#1}\xspace}
\newcommand{\dDynClass}[1]{\myclass{$\Delta$-Dyn#1}\xspace}

\newcommand  {\myproblem} [1] {\textsc{#1}}

\newcommand{\problemIndent}{\hspace{5mm}}
\newcommand  {\problemdescr} [3] {
    \vspace{3mm}
    \def\Name{#1}
    \def\Input{#2}
    \def\Question{#3}
     \problemIndent\begin{tabular}{r p{\columnWidth}r}%
      \textit{Problem:} & \myproblem{\Name} \\
      \textit{Input:} & \Input \\
      \textit{Question:} & \Question
     \end{tabular}
    \vspace{3mm}
    }

\newcommand  {\querydescr} [3] {
\vspace{3mm}
\def\Name{#1}
\def\Input{#2}
\def\Question{#3}
  \problemIndent\begin{tabular}{r p{\columnWidth}r}%
  \textit{Query:} & \myproblem{\Name} \\
  \textit{Input:} & \Input \\
  \textit{Question:} & \Question
  \end{tabular}
\vspace{3mm}
}

\newcommand  {\dynproblemdescr} [4] {
    \vspace{3mm}
    \def\Name{#1}
    \def\Input{#2}
    \def\Updates{#3}  
    \def\Question{#4}
    \problemIndent\begin{tabular}{r p{\columnWidth}r}%
      \textit{Query:} & \myproblem{\Name} \\
      \textit{Input:} & \Input \\
      \textit{Question:} & \Question
    \end{tabular}
    \vspace{3mm}
}
\newcommand{\dynProbDescr}[4]{\dynproblemdescr{#1}{#2}{#3}{#4}}

\newcommand  {\problem}[1] {\myproblem{#1}}

\newcommand{\dynProb}[1] {\myproblem{Dyn(#1)}}
\newcommand{\class}{\calC}

\newcommand  {\TIME}    {\myclass{TIME}}
\newcommand  {\DTIME}   {\myclass{DTIME}}
\newcommand  {\NTIME}   {\myclass{NTIME}}
\newcommand  {\ATIME}   {\myclass{ATIME}}
\newcommand  {\SPACE}   {\myclass{SPACE}}
\newcommand  {\DSPACE}   {\myclass{DSPACE}}
\newcommand  {\NSPACE}  {\myclass{NSPACE}}
\newcommand  {\coNSPACE}        {\myclass{coNSPACE}}

\newcommand     {\LOGCFL}     {\myclass{LOGCFL}}
\newcommand     {\LOGDCFL}     {\myclass{LOGDCFL}}
\newcommand     {\LOGSPACE}     {\myclass{LOGSPACE}}
\newcommand     {\NLOGSPACE}     {\myclass{NLOGSPACE}}
\newcommand     {\classL}   {\myclass{L}}
\newcommand     {\NL}   {\myclass{NL}}
\newcommand     {\coNL}   {\myclass{coNL}}
\renewcommand   {\P}    {\myclass{P}}
\newcommand     {\myP}    {\myclass{P}}
\newcommand     {\PTIME}    {\myclass{PTIME}}
\newcommand     {\NP}   {\myclass{NP}}
\newcommand     {\NPC}   {\myclass{NPC}}
\newcommand     {\PH}   {\myclass{PH}}
\newcommand     {\coNP} {\myclass{coNP}}
\newcommand     {\NPSPACE}      {\myclass{NPSPACE}}
\newcommand     {\PSPACE}       {\myclass{PSPACE}}
\newcommand     {\IP}   {\myclass{IP}}
\newcommand     {\POLYLOGSPACE} {\myclass{POLYLOGSPACE}}
\newcommand     {\DET}  {\myclass{DET}}
\newcommand     {\EXP}  {\myclass{EXP}}
\newcommand     {\NEXP}  {\myclass{NEXP}}
\newcommand     {\EXPTIME}  {\myclass{EXPTIME}}
\newcommand     {\TWOEXPTIME}  {\myclass{2-EXPTIME}}
\newcommand     {\TWOEXP}  {\myclass{2-EXP}}
\newcommand     {\NEXPTIME}  {\myclass{NEXPTIME}}
\newcommand     {\coNEXPTIME}  {\myclass{coNEXPTIME}}
\newcommand     {\EXPSPACE}  {\myclass{EXPSPACE}}
\newcommand     {\RP}   {\myclass{RP}}
\newcommand     {\RL}   {\myclass{RL}}
\newcommand     {\coRP} {\myclass{coRP}}
\newcommand     {\ZPP}  {\myclass{ZPP}}
\newcommand     {\BPP}  {\myclass{BPP}}
\newcommand     {\PP}   {\myclass{PP}}
\newcommand     {\NC}   {\myclass{NC}}
\newcommand     {\SAC}   {\myclass{SAC}}
\newcommand     {\ACC}   {\myclass{ACC}}
\newcommand     {\tc}   {\myclass{TC}}   %
\newcommand     {\PPoly}{\myclass{\mbox{P}/\mbox{Poly}}} %

\newcommand     {\FOarb}   {\myclass{FO(arb)}}

\newcommand     {\NLIN}   {\myclass{NLIN}}
\newcommand     {\DLIN}   {\myclass{DLIN}}

\newcommand  {\APTIME}   {\myclass{APTIME}}
\newcommand  {\ALOGSPACE}   {\myclass{ALOGSPACE}}

\newcommand{\FO}{\StaClass{FO}}
\newcommand{\MSO}[1][\quant]{\StaClass{MSO}}
\newcommand{\EMSO}{\StaClass{$\exists$MSO}}
\newcommand{\QFO}[1][\quant]{\StaClass{\ensuremath{#1}FO}}
\newcommand{\cQFO}[1][\quant]{\StaClass{\ensuremath{\overline{#1}}FO}}
\newcommand{\EFO}{\QFO[\exists^*]}
\newcommand{\AFO}{\QFO[\forall^*]}
\newcommand{\AEFO}{\StaClass{$\forall/\exists$FO}}
\newcommand{\CQ}[1][]{\StaClass{CQ}}
\newcommand{\UCQ}[1][]{\StaClass{UCQ}}
\newcommand{\CQneg}[1][]{\StaClass{CQ\ensuremath{^{\mneg}}}}
\newcommand{\UCQneg}[1][]{\StaClass{UCQ\ensuremath{^{\mneg}}}}
\newcommand{\Prop}{\StaClass{Prop}}
\newcommand{\QF}{\StaClass{QF}}
\newcommand{\PropCQ}{\StaClass{PropCQ}}
\newcommand{\PropUCQ}{\StaClass{PropUCQ}}
\newcommand{\PropCQneg}{\StaClass{PropCQ{\ensuremath{^{\mneg}}}}}
\newcommand{\PropUCQneg}{\StaClass{PropUCQ{\ensuremath{^{\mneg}}}}}

\newcommand{\mneg}{\neg} %

\newcommand{\DynTC}{\DynClass{TC}}

\newcommand{\DynProp}{\DynClass{Prop}}
\newcommand{\DynPropIA}[2]{\DynClass{Prop}(#1\text{-in},#2\text{-aux})}
\newcommand{\DynPropA}[1]{\DynClass{Prop}(#1\text{-aux})}
\newcommand{\DynPropI}[1]{\DynClass{Prop}(#1\text{-in})}
\newcommand{\DynProj}{\DynClass{Projections}}
\newcommand{\DynQF}{\DynClass{QF}}
\newcommand{\DynFO}{\DynClass{FO}}
\newcommand{\DIDynFO}{\myclass{DI-DynFO}\xspace}
\newcommand{\DynFOIA}[2]{\DynClass{FO}(#1\text{-in},#2\text{-aux})}
\newcommand{\DynFOA}[1]{\DynClass{FO}(#1\text{-aux})}
\newcommand{\DynFOI}[1]{\DynClass{FO}(#1\text{-in})}
\newcommand{\DynFOpos}{\DynClass{FO$^{\wedge, \vee}$}}
\newcommand{\DynFOand}{\DynClass{FO$^{\wedge}$}}

\newcommand{\DynC}{\DynClass{$\class$}}
\newcommand{\DynUCQ}{\DynClass{UCQ}}
\newcommand{\DynCQ}{\DynClass{CQ}}
\newcommand{\DynUCQneg}{\DynClass{UCQ$^\mneg$}}
\newcommand{\DynCQneg}{\DynClass{CQ$^\mneg$}}
\newcommand{\DynCQPM}{\DynCQneg}
\newcommand{\DyncQFO}{\DynClass{$\cquant$FO}}

\newcommand{\DynQFO}[1][\quant]{\DynClass{\QFO[#1]}}
\newcommand{\DynEFO}{\DynQFO[\exists^*]}
\newcommand{\DynAFO}{\DynQFO[\forall^*]}

\newcommand{\DynAEFO}{\DynClass{$\forall/\exists$FO}}
\newcommand{\DynAND}{\DynClass{PropCQ}}
\newcommand{\DynAnd}{\DynAND}
\newcommand{\DynPropCQ}{\DynAND}

\newcommand{\DynPropPos}{\DynClass{PropUCQ}}
\newcommand{\DynPropAO}{\DynPropPos}
\newcommand{\DynPropUCQ}{\DynPropPos}

\newcommand{\DynAndNeg}{\DynClass{PropCQ{\ensuremath{^{\mneg}}}}}
\newcommand{\DynPropCQneg}{\DynAndNeg}
\newcommand{\DynPropUCQneg}{\DynClass{PropUCQ{\ensuremath{^{\mneg}}}}}

\newcommand{\DynOrNeg}{\DynClass{Or{\ensuremath{^{\mneg}}}}}

\newcommand{\dDynProp}{\dDynClass{Prop}}
\newcommand{\dDynPropPos}{\dDynClass{PropUCQ}}
\newcommand{\dDynAndOr}{\dDynPropPos}
\newcommand{\dDynQF}{\dDynClass{QF}}
\newcommand{\dDynFO}{\dDynClass{FO}}
\newcommand{\dDynFOpos}{\dDynClass{FO$^{\wedge, \vee}$}}
\newcommand{\dDynFOand}{\dDynClass{FO$^{\wedge}$}}
\newcommand{\dDynC}{\dDynClass{$\class$}}
\newcommand{\dDynUCQ}{\dDynClass{UCQ}}
\newcommand{\dDynCQ}{\dDynClass{CQ}}
\newcommand{\dDynUCQneg}{\dDynClass{UCQ$^\mneg$}}
\newcommand{\dDynCQneg}{\dDynClass{CQ$^\mneg$}}
\newcommand{\dDynCQPM}{\dDynCQneg}

\newcommand{\dDynQFO}[1][\quant]{\dDynClass{\QFO[#1]}}
\newcommand{\dDynEFO}{\dDynQFO[\exists^*]}
\newcommand{\dDynAFO}{\dDynQFO[\forall^*]}

\newcommand{\dDynAEFO}{\dDynClass{$\forall/\exists$FO}}
\newcommand{\dDynAND}{\dDynPropCQ}
\newcommand{\dDynAnd}{\dDynAND}
\newcommand{\dDynConj}{\dDynClass{Conj}}
\newcommand{\dDynPropAO}{\dDynClass{Prop$^{\wedge, \vee}$}}
\newcommand{\dDyncQFO}{\dDynClass{$\cquant$FO}}
\newcommand{\dDynPropUCQneg}{\dDynClass{PropUCQ{\ensuremath{^{\mneg}}}}}
\newcommand{\dDynPropUCQ}{\dDynClass{PropUCQ}}
\newcommand{\dDynPropCQneg}{\dDynClass{PropCQ{\ensuremath{^{\mneg}}}}}
\newcommand{\dDynPropCQ}{\dDynClass{PropCQ}}

\newcommand{\equalcardinality}{\textsc{EqualCardinality}\xspace}%
\newcommand{\reach}{\textsc{Reach}\xspace}%
\newcommand{\altreach}{\textsc{Alt-Reach}\xspace}%

\newcommand{\stgraph}{$s$-$t$-graph\xspace}
\newcommand{\stgraphs}{$s$-$t$-graphs\xspace}
\newcommand{\reachQ}{\textsc{Reach}\xspace}
\newcommand{\streachQ}{\textsc{$s$-$t$-Reach}\xspace}
\newcommand{\streachabilityquery}{$s$-$t$-reachability query\xspace}
\newcommand{\stTwoPath}{\problem{$s$-$t$-Two\-Path}\xspace}
\newcommand{\sTwoPath}{\problem{$s$-Two\-Path}\xspace}
\newcommand{\clique}[1]{\problem{$#1$-Clique}\xspace}
\newcommand{\colorability}[1]{\problem{$#1$-Col}\xspace}
\newcommand{\streach}{$s$-$t$-Reach}
\newcommand{\streachp}{\problem{\streach}\xspace}
\newcommand{\layeredstreach}[1]{#1-Layered-$s$-$t$-Reach}
\newcommand{\layeredstreachp}[1]{\problem{\layeredstreach{#1}}\xspace}

\newcommand{\Emptiness}[1][]{\problem{Emptiness}\ifthenelse{\equal{#1}{}}{}{(#1)}\xspace}
\newcommand{\Consistency}[1][]{\problem{Consistency}\ifthenelse{\equal{#1}{}}{}{(#1)}\xspace}
\newcommand{\HI}[1][]{\problem{HistoryIndependence}\ifthenelse{\equal{#1}{}}{}{(#1)}\xspace}

\newcommand{\dynClique}[1]{\dynProb{$#1$-Clique}\xspace}
\newcommand{\dynColorability}[1]{\dynProb{$#1$-Col}\xspace}

\newcommand{\probEqualCardinalityText}{EqualCardinality}
\newcommand{\EqualCardinality}{\problem{\probEqualCardinalityText}\xspace}
\newcommand{\EqualCardinalityDescr}{\problemdescr{\probEqualCardinalityText}{Unary relations $A$ and $B$}{Do $A$ and $B$ have the same cardinality?\xspace}}

\newcommand{\dynEqualCardinality}{\dynProb{\probEqualCardinalityText}\xspace}
\newcommand{\dynEqualCardinalityDescr}{\dynProbDescr{\probEqualCardinalityText}{Unary relations $A$ and $B$}{Element insertions and deletions}{Do $A$ and $B$ have the same cardinality?\xspace}}

\newcommand{\dynReachQ}{\dynProb{\textsc{Reach}}\xspace}
\newcommand{\dynstReachQ}{\dynProb{\textsc{$s$-$t$-Reach}}\xspace}

\newcommand{\dynstTwoPath}{\dynProb{\stTwoPath}\xspace}
\newcommand{\dynsTwoPath}{\dynProb{\sTwoPath}\xspace}

\newcommand{\dynlayeredstreach}[1]{Dyn-#1-Layered-$s$-$t$-Reach}
\newcommand{\dynlayeredstreachp}[1]{\problem{\dynlayeredstreach{#1}}\xspace}

\makeatletter{}%
\newcommand{\mtext}[1]{\textsc{#1}}

\providecommand {\calA}      {{\mathcal A}\xspace}
\providecommand {\calB}      {{\mathcal B}\xspace}
\providecommand {\calC}      {{\mathcal C}\xspace}
\providecommand {\calD}      {{\mathcal D}\xspace}
\providecommand {\calE}      {{\mathcal E}\xspace}
\providecommand {\calF}      {{\mathcal F}\xspace}
\providecommand {\calG}      {{\mathcal G}\xspace}
\providecommand {\calH}      {{\mathcal H}\xspace}
\providecommand {\calK}      {{\mathcal K}\xspace}
\providecommand {\calI}      {{\mathcal I}\xspace}
\providecommand {\calL}      {{\mathcal L}\xspace}
\providecommand {\calM}      {{\mathcal M}\xspace}
\providecommand {\calN}      {{\mathcal N}\xspace}
\providecommand {\calO}      {{\mathcal O}\xspace}
\providecommand {\calP}      {{\mathcal P}\xspace}
\providecommand {\calQ}      {{\mathcal Q}\xspace}
\providecommand {\calR}      {{\mathcal R}\xspace}
\providecommand {\calS}      {{\mathcal S}\xspace}
\providecommand {\calT}      {{\mathcal T}\xspace}
\providecommand {\calU}      {{\mathcal U}\xspace}
\providecommand {\calV}      {{\mathcal V}\xspace}
\providecommand {\calX}      {{\mathcal X}\xspace}
\providecommand {\calZ}      {{\mathcal Z}\xspace}

\newcommand{\mhat}[1]{\widehat{#1}}

\newcommand{\Ah}{\mhat{A}}
\newcommand{\Bh}{\mhat{B}}
\newcommand{\Ch}{\mhat{C}}
\newcommand{\Dh}{\mhat{D}}
\newcommand{\Eh}{\mhat{E}}
\newcommand{\Fh}{\mhat{F}}
\newcommand{\Gh}{\mhat{G}}
\newcommand{\Hh}{\mhat{H}}
\newcommand{\Ih}{\mhat{I}}
\newcommand{\Jh}{\mhat{J}}
\newcommand{\Kh}{\mhat{K}}
\newcommand{\Lh}{\mhat{L}}
\newcommand{\Mh}{\mhat{M}}
\newcommand{\Nh}{\mhat{N}}
\newcommand{\Oh}{\mhat{O}}
\newcommand{\Ph}{\mhat{P}}
\newcommand{\Qh}{\mhat{Q}}
\newcommand{\Rh}{\mhat{R}}
\newcommand{\Sh}{\mhat{S}}
\newcommand{\Th}{\mhat{T}}
\newcommand{\Uh}{\mhat{U}}
\newcommand{\Vh}{\mhat{V}}
\newcommand{\Wh}{\mhat{W}}
\newcommand{\Xh}{\mhat{X}}
\newcommand{\Yh}{\mhat{Y}}
\newcommand{\Zh}{\mhat{Z}}

\newcommand{\Psih}{\mhat{\Psi}}
\newcommand{\psih}{\mhat{\psi}}
\newcommand{\Phih}{\mhat{\Phi}}
\newcommand{\phih}{\mhat{\phi}}
\newcommand{\varphih}{\mhat{\varphi}}

\newcommand{\N}{\ensuremath{\mathbb{N}}}

\newcommand{\Q}{\ensuremath{\mathbb{Q}}}

\newcommand{\R}{\ensuremath{\mathbb{R}}}

\newcommand{\perm}{\ensuremath{\pi}}

\newcommand{\allsubsets}[2]{[#1]^{#2}}

\newcommand{\pvec}[1]{\vec{#1}\mkern2mu\vphantom{#1}}

\newcommand{\kexp}[2]{\ensuremath{\exp^{#1}\hspace{-0.5mm}(#2)}}

\newcommand{\tower}[2]{\ensuremath{\text{tow}_{#1}\hspace{-0.5mm}(#2)}}

\newcommand{\klog}[2]{\ensuremath{\log^{#1}{\hspace{-0.5mm}(#2)}}}

\newcommand{\disjointunion}{\uplus}

\providecommand{\power}[1]{\ensuremath{\calP(#1)}\xspace}

\newcommand{\restrict}[2]{#1\mspace{-3mu}\upharpoonright \mspace{-3mu}#2}

\newcommand{\isomorph}{\simeq}
\newcommand{\isomorphVia}[1]{\isomorph_{#1}}
\newcommand{\swap}[2]{id{[#1, #2]}}

\newcommand{\df}{\ensuremath{\mathrel{\smash{\stackrel{\scriptscriptstyle{
    \text{def}}}{=}}}} \;}

\newcommand{\refeq}[1]{\ensuremath{{\stackrel{\scriptstyle{
    \text{#1}}}{=}}}}

\newcommand{\longlongeq}{=\joinrel=\joinrel=\joinrel=}
\newcommand{\longeq}{=\joinrel=\joinrel=}
\newcommand{\reflongeq}[1]{\ensuremath{{\stackrel{\scriptstyle{
    \text{#1}}}{\longeq}}}}

\newcommand{\ramseyw}[1]{\ensuremath{R_{#1}}}

\makeatletter %
\newcommand{\auxramsey}[4]{
  \@ifmtarg{#1}{
    \@ifmtarg{#4}{
      \ensuremath{R(#2; #3)}
    }{
      \ensuremath{R^#4(#2; #3)}
    }
   }{
    \@ifmtarg{#4}{
      \ensuremath{R_{#1}(#2; #3)}
    }{
      \ensuremath{R^#4_{#1}(#2; #3)}
    }
  }
}

\newcommand{\ramsey}[3]{\auxramsey{#1}{#2}{#3}{}}
\newcommand{\homramsey}[2]{\auxramsey{}{#2}{#1}{\text{hom}}}
\newcommand{\mfoldramsey}[3]{\auxramsey{}{#2}{#1}{#3}}

\newcommand{\norder}{\prec}

\newcommand{\col}{col}

\newcommand{\property}{($\ast$)}

\newcommand{\subseq}{\sqsubseteq}

\newcommand{\derive}{\Rightarrow}
\newcommand{\rmapsto}{\rightarrow}

\newcommand{\lpath}[1][]{\ensuremath{\mathrel{\smash{\stackrel{\scriptscriptstyle{
    #1}}{\rightsquigarrow}}}}}

\makeatletter{}%

   \theoremstyle{plain}

   \newtheorem{proposition}[theorem]{Proposition}

     \newtheorem{goal}{Goal}
    \theoremstyle{definition}

    \newtheorem {openquestion}{Open question}
    \newtheorem {question}{Question}
    \newtheorem {mainquestion}{Main question}

    \newenvironment{proofsketch}{\begin{proof}[Proof sketch.]}{\end{proof}}
    \newenvironment{proofidea}{\begin{proof}[Proof idea.]}{\end{proof}}

    \newenvironment{proofof}[1]{\begin{proof}[Proof (of #1).]}{\end{proof}}
   \newenvironment{proofsketchof}[1]{\begin{proof}[Proof sketch (of #1).]}{\end{proof}}
      \newenvironment{proofideaof}[1]{\begin{proof}[Proof idea (of #1).]}{\end{proof}}

\newenvironment{proofenum}{\begin{enumerate}[label=(\alph*),wide=0pt, listparindent=15pt]}{\end{enumerate}}

\makeatletter{}%
\newcommand{\inc}{\ensuremath{\text{inc}}}
\newcommand{\dec}{\ensuremath{\text{dec}}}
\newcommand{\transfer}{\ensuremath{\text{transfer}}}
\newcommand{\ifzero}{\ensuremath{\text{ifzero}}}
\newcommand{\ifzdec}{\ensuremath{\text{ifzdec}}}
\newcommand{\ifz}{\ifzero}

\newcommand{\eval}[3]{#1(#2/#3)}

\newcommand{\assignment}{\theta}

\newcommand{\arity}{\ensuremath{\text{Ar}}}
\newcommand{\arityFun}{\ensuremath{Ar_{\text{fun}}}}

\newcommand{\schema}{\ensuremath{\tau}\xspace}
\newcommand{\schemah}{\hat{\schema}}
\newcommand{\relSchema}{\schema_{\text{rel}}}
\newcommand{\relSchemah}{\schemah_{\text{rel}}}
\newcommand{\conSchema}{\schema_{\text{const}}}
\newcommand{\conSchemah}{\schemah_{\text{const}}}
\newcommand{\funSchema}{\schema_{\text{fun}}}
\newcommand{\funSchemah}{\schemah_{\text{fun}}}
\newcommand{\Terms}[2]{\textsc{Terms}^{#2}_{#1}} 

\newcommand{\struc}{\calS}
\newcommand{\struca}{\struc}
\newcommand{\strucb}{\calT}

\newcommand{\unaryTypes}[1]{\mathcal{UN}_{#1}}
\newcommand{\binaryTypes}[1]{\mathcal{BIN}_{#1}}
\newcommand{\naryTypes}[2]{\mathfrak{T}_{#1,#2}}

\newcommand{\nb}[3]{\calN_{#2}^{#3}(#1)}
\newcommand{\nbv}[3]{\vec \calN_{#2}^{#3}(#1)}

\newcommand{\inpTypes}{\ensuremath{\Gamma_{\text{in}}}\xspace}
\newcommand{\auxTypes}{\ensuremath{\Gamma_{\text{aux}}}\xspace}

\newcommand{\qfrank}[1]{\ensuremath{\text{rank-}#1}}

\newcommand{\mthen}{\rightarrow}
\newcommand{\mand}{\wedge}
\newcommand{\mor}{\vee}
\newcommand{\munion}{\cup}
\newcommand{\mintersect}{\cap}
\newcommand{\mdisjunion}{\biguplus}
\newcommand{\sem}[2]{\ensuremath{\llbracket #1\rrbracket_{#2}}} %

\newcommand{\arb}{\ensuremath{\star}}%
\newcommand{\generic}{\textsc{generic}}
\newcommand{\quant}{\mathbb{Q}}
\newcommand{\cquant}{\overline{\mathbb{Q}}}

\newcommand{\nd}{d}
\newcommand{\formulas}{\calC}
\newcommand{\symneg}[1]{\widehat{#1}}

\newcommand{\type}[2]{\ensuremath{\langle #1, #2 \rangle}}
\newcommand{\stype}[3]{\ensuremath{\langle #1, #2 \rangle_{#3}}}

\newcommand{\behaveEqual}[1]{\approx_{#1}}

\newcommand{\types}[2]{types_{#1}(#2)}
\newcommand{\numTypes}[2]{|\types{{#1}}{#2}|}

\newcommand{\eqtype}{\epsilon}

\newcommand{\db}{\ensuremath{\calD}\xspace}
\newcommand{\inp}{\ensuremath{\calI}\xspace}
\newcommand{\aux}{\ensuremath{\calA}\xspace}
\newcommand{\builtin}{\ensuremath{\calB}\xspace}
\newcommand{\domain}{\ensuremath{ D}\xspace}
\newcommand{\actDomain}{\ensuremath{ D_\text{act}}\xspace}
\newcommand{\emptyDB}{\ensuremath{\db_\emptyset}\xspace}

\newcommand{\query}{\ensuremath{\calQ}}
\newcommand{\cq}{\calC}

\newcommand{\querys}{\ensuremath{R_\query}\xspace}
\newcommand{\bitquerys}{\ensuremath{\Acc}\xspace}

\newcommand{\ans}[2]{\mtext{ans}(#1, #2)}

\newcommand{\updates}{\ensuremath{\Delta}}
\newcommand{\abstrDel}{\ensuremath{\updates_{Del}}}
\newcommand{\abstrIns}{\ensuremath{\updates_{Ins}}}
\newcommand{\abstrUpd}{\ensuremath{\updates}}

\newcommand{\init}{\mtext{Init}\xspace}

\newcommand{\ins}{\mtext{ins}\xspace}
\newcommand{\del}{\mtext{del}\xspace}

\newcommand{\insertdescr}[2]{\textbf{Insertion of \ensuremath{#2} into \ensuremath{#1}.}}
\newcommand{\deletedescr}[2]{\textbf{Deletion of \ensuremath{#2} from \ensuremath{#1}.}}

\newcommand{\state}{\ensuremath{\struc}\xspace}

\newcommand{\inpSchema}{\ensuremath{\schema_{\text{in}}}\xspace}
\newcommand{\auxSchema}{\ensuremath{\schema_{\text{aux}}}\xspace}
\newcommand{\eqSchema}{\ensuremath{\schema_{=}}\xspace}
\newcommand{\builtinSchema}{\ensuremath{\schema_{\text{bi}}}\xspace}

\newcommand{\auxInit}{\init_{\text{aux}}}
\newcommand{\builtinInit}{\init_{\text{bi}}}

\newcommand{\upProg}{\ensuremath{P}\xspace}
\newcommand{\prog}{\ensuremath{\calP}\xspace}
\newcommand{\progb}{\ensuremath{Q}\xspace}

\newcommand{\updateDB}[2]{\ensuremath{#1(#2)}}
\newcommand{\updateState}[3][\prog]{\ensuremath{#1_{#2}(#3)}}
\newcommand{\updateStateI}[3][\prog]{\ensuremath{#1_{#2}(#3)}}
\newcommand{\updateRelation}[4]{\restrict{\ensuremath{{#1}_{#2}(#3)}}{#4}}

\newcommand{\transition}[3]{\ensuremath{{#1} \xrightarrow{#2}{#3}}}

\makeatletter %
\newcommand{\uf}[4]{
  \@ifmtarg{#4}{
    \ensuremath{\phi^{#1}_{#2}(#3)}
   }{
    \ensuremath{\phi^{#1}_{#2}(#3; #4)}
  }
}
\newcommand{\huf}[4]{
  \@ifmtarg{#4}{
    \ensuremath{\widehat{\phi}^{#1}_{#2}(#3)}
   }{
    \ensuremath{\widehat{\phi}^{#1}_{#2}(#3; #4)}
  }
}

\newcommand{\ufb}[4]{
  \@ifmtarg{#4}{
    \ensuremath{\psi^{#1}_{#2}(#3)}
   }{
    \ensuremath{\psi^{#1}_{#2}(#3; #4)}
  }
}

\newcommand{\ufbwa}[2]{
  \ensuremath{\psi^{#1}_{#2}}
}

\newcommand{\ufwa}[2]{
  \ensuremath{\phi^{#1}_{#2}}
}

  \makeatletter %
  \newcommand{\ufsubstitute}[5]{
    \@ifmtarg{#5}{
      \ensuremath{\phi^{#2}_{#3}[#1](#4)}
    }{
      \ensuremath{\phi^{#2}_{#3}[#1](#4; #5)}
    }
  }

  \makeatletter %
  \newcommand{\ufsubstitutewa}[3]{
      \ensuremath{\phi^{#2}_{#3}[#1]}
  }
  \makeatletter %
  \newcommand{\substitutewa}[2]{
      \ensuremath{#1[#2]}
  }

\newcommand{\ut}[4]{
  \@ifmtarg{#4}{
    \ensuremath{t^{#1}_{#2}(#3)}
   }{
    \ensuremath{t^{#1}_{#2}(#3; #4)}
  }
}

\newcommand{\utwa}[2]{\ensuremath{t^{#1}_{#2}}}
\newcommand{\ite}[3]{
  \@ifmtarg{#1}{
    \ensuremath{\mtext{ITE}}
   }{
    \mtext{ITE}(#1,#2,#3)  
  }
}

\makeatletter{}%

\providecommand{\nc}{\newcommand}
\providecommand{\rnc}{\renewcommand}
\providecommand{\pc}{\providecommand}

\renewcommand{\labelenumi}{(\alph{enumi})}

\newcommand{\Erdos}{Erd\H{o}s}

\ifcomments
\nc{\commentbox}[1]{\noindent\framebox{\parbox{\linewidth}{#1}}}
\nc{\todo}[1]{\ \\ {\color{red} \fbox{\parbox{\linewidth}{{\sc
          ToDo}:\\  #1}}}}

\setlength{\marginparwidth}{2.5cm}
\setlength{\marginparsep}{3pt}

\newcounter{CommentCounter}
\newcommand{\acomment}[2]{\ \\ \fbox{\parbox{\linewidth}{{\sc #1}: #2}}}
\newcommand{\mcomment}[2]{{\color{blue}(#1)}\footnote{#1: #2}} %
\else
\nc{\commentbox}[1]{}
\newcommand{\mcomment}[2]{}
\newcommand{\acomment}[2]{}
\fi

\ifchanges

\newcommand{\loldnew}[3]{\commentbox{{\textcolor{blue}{\setlength{\fboxsep}{1pt}\fbox{\small
          #1}}} \textcolor{red}{\footnotesize #2}}
  \textcolor{blue}{#3}}
\setul{}{0.2mm}
\setstcolor{red}
\newcommand{\oldnew}[3]{{\textcolor{blue}{\setlength{\fboxsep}{1pt}\fbox{\small
        #1}}} \st{\footnotesize #2} \textcolor{blue}{#3}}

\else
\newcommand{\loldnew}[3]{#3}
\newcommand{\oldnew}[3]{#3}
\fi

 \nc{\tzm}[1]{\mcomment{TZ}{#1}}
 \nc{\tsm}[1]{\mcomment{TS}{#1}}
 \nc{\nilsm}[1]{\mcomment{NV}{#1}}
 
 \nc{\tz}[1]{\acomment{TZ}{#1}}
 \nc{\thz}[1]{\acomment{TZ}{#1}}
 \nc{\ts}[1]{\acomment{TS}{#1}}
 \nc{\nils}[1]{\acomment{NV}{#1}}

\nc{\tzon}[2][]{\oldnew{TZ}{#1}{#2}} 
\nc{\tson}[2][]{\oldnew{TS}{#1}{#2}}
\nc{\nilson}[2][]{\oldnew{NV}{#1}{#2}}

\nc{\tzlon}[2][]{\loldnew{TZ}{#1}{#2}} 
\nc{\tslon}[2][]{\loldnew{TS}{#1}{#2}}
\nc{\nilslon}[2][]{\loldnew{NV}{#1}{#2}}
\makeatletter{}%
\newcommand{\apptheoremtitlefont}[1]{\textbf{#1}}
\newcommand{\apptheoremcontentfont}{\itshape}

\newcommand{\apponlystartmarker}{ $\blacktriangleright\blacktriangleright\blacktriangleright$ }
\newcommand{\apponlyendmarker}{ $\blacktriangleleft\blacktriangleleft\blacktriangleleft$ }
\newcommand{\apprepetitionstartmarker}{}%
\newcommand{\apprepetitionendmarker}{}%

\newcommand{\initialAppendix}{
  \section*{Appendix}

  In the appendix we give the proofs that have been omitted in the main text. For proofs that are partially present in the main article, we repeat the full proof and its context. Parts that are only repeated are marked by \apprepetitionstartmarker and \apprepetitionendmarker. For the convenience of the reader we repeat the full Section \ref{section:hi}.  
  
}
  
\newcommand{\writeAppendix}{
  
  \initialAppendix
}
\newcommand{\toAppendix}[1]{
  \makeatletter
   \g@addto@macro\writeAppendix{#1}
  \makeatother
}

\newcommand{\toMainAndAppendix}[1]{
  \longVersion{#1}
  \shortVersion{
    #1%
    \toAppendix{%
	\apprepetition{#1} \par
    }
  }
}

\newcommand{\toLongAndAppendix}[1]{
  \longVersion{#1}%
  \shortVersion{    
    \toAppendix{%
	#1 \par
    }
  }
}

\newcommand{\atheorem}[2]{
  \begin{theorem}\label{#1}%
    #2%
  \end{theorem}%
  \toAppendix{%
    \begin{apptheorem}{\ref{#1}}{}
      #2%
    \end{apptheorem}%
  }
}

\newcommand{\alemma}[2]{
  \begin{lemma}\label{#1}%
    #2%
  \end{lemma}%
  \toAppendix{%
    \begin{applemma}{\ref{#1}}{}%
      #2%
    \end{applemma}%
  }
}

\newcommand{\aproposition}[2]{
  \begin{proposition}\label{#1}%
    #2%
  \end{proposition}%
  \toAppendix{%
    \begin{appproposition}{\ref{#1}}{}%
      #2%
    \end{appproposition}%
  }
}

\newcommand{\aproof}[3]{%
  \longVersion{
    \begin{proof}%
      #1%
      #3%
    \end{proof}%
  }
  \shortVersion{
    \@ifmtarg{#2}{}{%
      \begin{proof}%
	#1%
	#2%
      \end{proof}%
    }
    \toAppendix{%
      \begin{proof}%
	\@ifmtarg{#1}{}{\apprepetition{#1}} \par
	#3
      \end{proof}%
    }
  }
}

\newcommand{\aproofsketch}[3]{%
  \longVersion{
    \begin{proof}%
      #1%
      #3%
    \end{proof}%
  }

  \shortVersion{
    \@ifmtarg{#2}{}{%
      \begin{proofsketch}%
	#1%
	#2%
      \end{proofsketch}%
    }
    \toAppendix{%
      \begin{proof}%
	\@ifmtarg{#1}{}{\apprepetition{#1}} \par
	#3
      \end{proof}%
    }
  }
}

\newcommand{\aproofidea}[3]{%
  \longVersion{
    \begin{proof}%
      #1%
      #3%
    \end{proof}%
  }

  \shortVersion{
    \@ifmtarg{#2}{}{%
      \begin{proofidea}%
	#1%
	#2%
      \end{proofidea}%
    }
    \toAppendix{%
      \begin{proof}%
	\@ifmtarg{#1}{}{\apprepetition{#1}} \par
	#3
      \end{proof}%
    }
  }
}

\newcommand{\shortOrLong}[2]{%
  \shortVersion{#1}%
  \longVersion{#2}%
}

\makeatletter
\newcommand{\theoremcont}[3]{
   \def\Type{#1}
   \def\Number{#2}
   \def\Label{#3}
  \@ifmtarg{#3}{
     \apptheoremtitlefont{\Type\ \Number.} \apptheoremcontentfont
   }{
    \apptheoremtitlefont{\Type\ \Number}\ \apptheoremcontentfont(\Label).
  }
}

\newenvironment{applemma}[2]{\vspace{2mm}\par\theoremcont{Lemma}{#1}{#2}}{\vspace{0mm}\par}
\newenvironment{apptheorem}[2]{\vspace{2mm}\par\theoremcont{Theorem}{#1}{#2}}{\vspace{2mm} \par }
\newenvironment{appcorollary}[2]{\theoremcont{Corollary}{#1}{#2}}{\vspace{2mm}}
\newenvironment{appproposition}[2]{\theoremcont{Proposition}{#1}{#2}}{\vspace{2mm}}
\newenvironment{appdefinition}[2]{\theoremcont{Definition}{#1}{#2}}{\vspace{2mm}}
\newenvironment{appexample}[1]{\vspace{2mm}\textit{Example #1.}}{\vspace{2mm}}

\newcommand{\apponlystart}{
  \apponlystartmarker
}
\newcommand{\apponlyend}{
  \apponlyendmarker
}

\newcommand{\apprepetition}[1]{
  \apprepetitionstartmarker #1 \apprepetitionendmarker
}

\newcommand{\powsym}{^\wedge}

\newcommand{\progToGraph}[1]{\ensuremath{{\langle #1 \rangle}}}
\newcommand{\progToString}[1]{\ensuremath{{\langle #1 \rangle}}}
\newcommand{\toString}[1]{\ensuremath{{\langle #1 \rangle}}}
\newcommand{\toStructure}[1]{\ensuremath{{\langle #1 \rangle}}}
\newcommand{\progToGraphInv}[1]{\ensuremath{{\langle #1 \rangle^{-1}}}}

\newcommand{\progToGraphPadded}[1]{\ensuremath{{\langle\langle #1 \rangle\rangle}}}
\newcommand{\progToGraphPaddedInv}[1]{\ensuremath{{\langle\langle #1 \rangle\rangle^{-1}}}}

\newcommand{\LineIf}[2]{\State \algorithmicif\ {#1}\ \algorithmicthen\ {#2}}

\algnewcommand\algorithmicinput{\textbf{Input:}}
\algnewcommand\INPUT{\item[\algorithmicinput]}

\algnewcommand\algorithmicoutput{\textbf{Output:}}
\algnewcommand\OUTPUT{\item[\algorithmicoutput]}

\newcommand{\columnWidth}{9cm}
\renewcommand{\problemIndent}{\hspace{0mm}}

\newcommand{\substruclemma}{Substructure Lemma\xspace}
\newcommand{\First}{\mtext{First}}
\newcommand{\List}{\mtext{List}}
\newcommand{\Last}{\mtext{Last}}
\newcommand{\In}{\mtext{In}}
\newcommand{\Out}{\mtext{Out}}
\newcommand{\Empty}{\mtext{Empty}}
\newcommand{\RelName}[1]{\mtext{#1}}
\newcommand{\Acc}{\mtext{Acc}}

\newcommand{\Odd}{\mtext{Odd}}
\newcommand{\odd}{\text{odd}}
\newcommand{\even}{\text{even}}

\newcommand{\Counter}{\mtext{Counter}}
\newcommand{\isEmpty}{\mtext{Empty}}
\newcommand{\Zero}{\mtext{Zero}}

\newcommand{\congruent}[2]{\sim_{#1, #2}}

\newcommand{\Appendix}{\centerline{------------ Material for Appendix ---------}}

\newcommand{\lhi}{locally history independent\xspace}

\newcommand{\phibad}{\ensuremath{\varphi_{\text{bad}}}}
   \newcommand{\shortVersion}[1]{} \newcommand{\longVersion}[1]{#1}

  \maketitle
  \begin{abstract}
    A dynamic program, as introduced by Patnaik and Immerman (1994), maintains the result of a fixed query for an
    input database which is subject to tuple insertions and
    deletions. It can use an auxiliary database whose relations are updated via first-order formulas upon modifications of the input database. %
    
This paper studies static analysis problems for dynamic programs and
investigates, more specifically, the decidability of the following
three questions. Is the
answer relation of a given dynamic program always empty? Does a
program actually maintain a query? 
Is the content of auxiliary relations independent of the modification sequence that lead to an input database?
   In general, all these problems can easily be seen to be undecidable
   for full first-order programs. Therefore the paper aims at
   pinpointing the exact decidability borderline for programs
   with restricted arity (of the input and/or auxiliary database) and
   restricted quantification. 
  \end{abstract}

  \section{Introduction}\label{section:introduction}

\makeatletter{}%
In modern database scenarios data is subject to frequent changes. In order to avoid costly re-computation of queries from scratch after each small modification of the data, one can try to use previously computed auxiliary data. This auxiliary data then needs to be updated dynamically whenever the database changes.

The descriptive dynamic complexity framework (short: dynamic complexity) by Patnaik and Immerman \cite{PatnaikI94} models this setting from a declarative perspective. It was mainly inspired by updates in relational databases. Within this framework, for a relational database subject to change, a \emph{dynamic program} maintains auxiliary relations with the intention to help answering a query \query. When a modification to the database, that is an insertion or deletion of a tuple, occurs, every auxiliary relation is updated through a first-order update formula (or, equivalently, through a core SQL query) that can refer to the database as well as to the auxiliary relations. The result of $\query$ is, at every time, represented by some distinguished auxiliary relation. The class of all  queries maintainable by  dynamic programs with first-order update formulas is called \DynFO and we refer to such programs as \DynFO-programs.
We note that shortly before the work of Patnaik and Immerman, the declarative approach was independently formalized in a similar way by Dong, Su and Topor \cite{DongST95}.

The main question studied in Dynamic Complexity has been which queries that are not statically expressible in first-order logic (and therefore not in Core SQL), can be maintained by \DynFO-programs. Recently, it has been shown that the Reachability query, a very natural such query, can be maintained by \DynFO programs~\cite{DattaKMSZ15}.
Altogether, research in Dynamic Complexity succeeded in proving that many non-FO queries are maintainable in \DynFO. These results and their underlying techniques yield many interesting insights into the the nature of Dynamic Complexity. 

However, to complete the understanding of Dynamic Complexity,  it would be desirable to complement these techniques by  methods for proving that certain queries are \emph{not} maintainable by \DynFO programs. But the state of the art with respect to inexpressibility results is much less favorable: at this point, no general techniques for showing that a query is not expressible in \DynFO  are available. In order to get a better overall picture of Dynamic Complexity in general and to develop methods for inexpressibility proofs in particular, various restrictions of \DynFO have been studied, based on, e.g., arity restrictions for the auxiliary relations \cite{DongLW95, DongS98, DongLW03}, fragments of first-order logic \cite{Hesse03, GeladeMS12, ZeumeS15, Zeume14}, or by other means \cite{DongS97, GraedelS12}. 

At the heart of our difficulties to prove inexpressibility results in Dynamic Complexity is our limited understanding of what dynamic programs with or without restrictions ``can do'' in general, and our limited ability to analyze what a particular dynamic program at hand ``does''. In this paper, we initiate a systematic study of the ``analyzability'' of dynamic programs.   Static analysis of queries has a long tradition in Database Theory and we follow this tradition by first studying the emptiness problem for dynamic programs, that is the question, whether there exists an initial database and a modification sequence that is accepted by a given dynamic program.\footnote{The exact framework will be defined in Section \ref{section:setting}, but we already mention that we will consider the setting in which databases are initially empty and the auxiliary relations are defined by first-order formulas.} Given the well-known undecidability of the finite satisfiability problem for first-order logic \cite{Trahtenbrot63}, it is not surprising that emptiness of \DynFO programs is undecidable in general. However, we try to pinpoint the borderline of undecidability for fragments of \DynFO based on restrictions of the arity of input relations, the arity of auxiliary relations and for the class \DynProp of programs with quantifier-free update formulas.

In the fragments where undecidability of emptiness does not directly follow from undecidability of satisfiability in the corresponding fragment of first-order logic, our undecidability proofs make use of dynamic programs whose query answer might not only depend on the database yielded by a certain modification sequence, but also on the sequence itself, that is, on the order in which tuples are inserted or (even) deleted. From a useful dynamic program one would, of course, expect that it is \emph{consistent} in the sense that its query answer always only depends on the current database, but not on the specific modification sequence by which it has been obtained. 
It turns out that the emptiness problem for consistent programs is easier than the general emptiness problem for dynamic programs. More precisely, there are fragments of \DynFO, for which an algorithm can decide emptiness for dynamic programs that come with a ``consistency  guarantee'', but for which  the emptiness problem is undecidable, in general. However, it turns out that the combination of a consistency test with an emptiness test for consistent programs does not gain any advantage over ``direct'' emptiness tests, since the consistency problem turns out to be as difficult as the general emptiness problem. 

Finally, we study a property that many dynamic programs in the literature share: they are \emph{history independent} in the sense that all  auxiliary relations always only depend on the current (input) database. History independence can be seen as a strong form of consistency in that it not only requires the query relation, but \emph{all} auxiliary relations to be determined by the input database. 
 History independent dynamic programs (also called \emph{memoryless} \cite{PatnaikI94} or \emph{deterministic} \cite{DongS97}) are still expressive enough to maintain interesting queries like undirected reachability \cite{GraedelS12}. But also some inexpressibility proofs have been found for such programs \cite{DongS97,GraedelS12, ZeumeS15}.
We study the \emph{history independence problem}, that is, whether a given dynamic program is history independent. In a nutshell, the history independence problem is the ``easiest'' of the static analysis problems considered in this paper. 

Our results, summarized in Table \ref{tab:results},  shed light on the borderline between decidable and undecidable fragments of \DynFO with respect to emptiness (and consistency), emptiness for consistent programs and history independence. While the picture is quite complete for the emptiness problem for general dynamic programs, for some fragments of~\DynProp there remain open questions regarding the emptiness problem for consistent dynamic programs and the history-independence problem.
Some of the results shown in this paper have been already presented in the master thesis of Nils Vortmeier \cite{Vortmeier13}. 

\begin{table}[t!]
  \centering
\begin{tabular}{l|C{3.2cm}|C{3.2cm}|C{3.2cm}}
  & Emptiness \newline Consistency &  Emptiness for consistent programs  & History\newline Independence\\
\hline
 Undecidable &  $\DynFOIA{1}{0}$\newline$\DynPropIA{2}{0}$\newline$\DynPropIA{1}{2}$
                                   & $\DynFOIA{1}{2}$\newline $\DynFOIA{2}{0}$ & $\DynFOIA{2}{0}$\\
\hline
 Decidable&$\DynPropIA{1}{1}$ & $\DynFOIA{1}{1}$\newline$\DynPropI{1}$\newline$\DynPropA{1}$ & $\DynFOI{1}$\newline$\DynPropA{1}$\\
\hline
 Open & &  $\DynPropIA{2}{2}$ and beyond &  $\DynPropIA{2}{2}$ and beyond\\
\end{tabular}
\caption{Summary of the results of this paper. $\DynFOIA{\ell}{m}$ stands for \DynFO-programs with (at most) $\ell$-ary input relations and $m$-ary auxiliary relations. $\DynFOA{m}$ and $\DynFOI{\ell}$ represent programs with $m$-ary auxiliary relations (and arbitrary input relations) and programs with $\ell$-ary input relations, respectively. Likewise for $\DynProp$.
}
\label{tab:results}
\end{table}

\subparagraph*{Outline}
We recall some basic definitions in Section \ref{section:preliminaries} and introduce the formal setting in Section \ref{section:setting}. 
The emptiness problem is defined and studied in Section \ref{section:emptiness}, where we first consider general dynamic programs (Subsection \ref{section:emptinessgeneral}) and then consistent dynamic programs (Subsection \ref{section:emptinessconsistent}).
In Subsection \ref{section:emptinessbuiltin} we briefly discuss the impact of built-in orders to the results.
The Consistency and History Independence problems are studied in  Sections \ref{section:consistency} and \ref{section:hi}, respectively. We conclude in Section \ref{section:conclusion}. \shortOrLong{Due to the space limit we only give proof sketches or even proof ideas in the body of this paper. Complete proofs can be found in the long version \cite{}.}{}

  \section{Preliminaries}\label{section:preliminaries}
\makeatletter{}%
We presume that the reader is familiar with basic notions from Finite Model Theory and refer to  \cite{EbbinghausFlum95, Libkin04} for a detailed introduction into this field. We review some basic definitions in order to fix notations.

In this paper, a \textit{domain} is a non-empty finite set. For tuples $\vec a = (a_1, \ldots, a_k)$ and $\vec b = (b_1, \ldots, b_\ell)$ over some domain~$\domain$, the $(k + \ell)$-tuple obtained by concatenating  $\vec a$ and $\vec b$ is denoted by $(\vec a, \vec b)$. 

A (relational) \emph{schema} is a collection $\schema$ of relation symbols\footnote{For simplicity we do not allow constants in this work but note that our results hold for relational schemas with constants as well.}
together with an arity function $\arity: \schema \rightarrow \N$. A \emph{database} $\db$ with schema $\schema$ and domain $\domain$ is a mapping that assigns to every relation symbol $R \in \schema$ a relation of arity $\arity(R)$ over $\domain$. The \emph{size of a database}, usually denoted by $n$, is the size of its domain. We call a database \emph{empty}, if all its relations are empty. We emphasize that empty databases have  non-empty domains.
 A  $\schema$-\emph{structure} $\struc$ is a pair $(\domain, \db)$ where $\db$ is a database with schema $\schema$ and domain $\domain$. Often we omit the schema when it is clear from the context. %

We write $\struc\models\varphi(\vec a)$ if the first-order formula $\varphi(\vec x)$ holds in $\struc$ under the variable assignment that maps $\vec x$ to $\vec a$.
The \emph{quantifier depth} of a first-order formula is the maximal nesting depth of quantifiers. 
 The \emph{rank-$q$ type} of a tuple $(a_1, \ldots, a_m)$ with respect to a $\schema$-structure $\struc$ is the set of all first-order formulas $\varphi(x_1, \ldots, x_m)$ (with equality) of quantifier depth at most $q$, for which $\struc\models\varphi(\vec a)$ holds.  By $\struc \equiv_q \struc'$ we denote that two structures $\struc$ and $\struc'$ have the same rank-$q$ type (of length 0 tuples).

For a subschema $\tau'\subseteq \tau$, the rank-$q$ $\tau'$-type of a tuple $\vec a$ in a $\tau$-structure \state is its  rank-$q$ type in the $\tau'$-reduct of \state.

We refer to the rank-0 type of a tuple also as its \emph{atomic type} and, since we mostly deal with rank-0 types, simply as its \emph{type}.  %
The \emph{equality type} of a tuple is the atomic type with respect to the empty schema.

The \emph{$k$-ary type} of a tuple $\vec a$ in a structure $\struc$ is its $\tau_{\le k}$-type, where $\tau_{\le k}$ consists of all relation symbols of $\tau$ with arity at most $k$.
The \emph{$\tau'$-color} of an element $a$ in $\struc$, for a subschema $\tau'$ of the schema of $\struc$, is its $\tau'_{1}$-type, where $\tau'_{1}$ consists of all unary relation symbols of $\tau'$. We often enumerate the possible $\tau'$-colors as $c_0,\ldots,c_L$, for some $L$ with $c_0$ being the color of elements that are in neither of the unary relations. We call these elements \emph{$\tau'$-uncolored}. If $\tau'$ is clear from the context we simply speak of colors and uncolored elements.

  \section{The dynamic complexity setting}\label{section:setting}
\makeatletter{}%

For a database \db over schema \schema, a \emph{modification} $\delta=(o,\vec a)$ consists of an operation $o\in \{\ins_S, \del_S\mid S\in\tau\}$ and a tuple $\vec a$ of elements from the domain of \db. By $\updateDB{\delta}{\db}$ we denote the result of applying $\delta$ to $\db$ with the obvious semantics of inserting or deleting the tuple $\vec a$ to or from relation $S^{\db}$. For a sequence $\alpha = \delta_1 \cdots \delta_N$  of modifications to a database $\db$ we let $\updateDB{\alpha}{\db} \df \updateDB{\delta_N}{\cdots (\updateDB{\delta_1}{\db})\cdots}$.

A \emph{dynamic instance}\footnote{The following introduction to dynamic descriptive complexity is similar to previous \mbox{work \cite{ZeumeS15, ZeumeS14icdt}}. } of a query $\query$ is a pair $(\db, \alpha)$, where $\db$ is a database over a domain $\domain$ and $\alpha$ is a sequence of modifications to $\db$. The dynamic query $\dynProb{$\query$}$ yields the result of  evaluating the query $\query$ on $\alpha(\db)$.

Dynamic programs, to be defined next, consist of an initialization mechanism and an update program.  The former  yields, for every (initial) database $\db$,  an initial state with initial auxiliary  data. The latter defines the new state of the dynamic program for each possible modification $\delta$.

A \emph{dynamic schema} is a pair $(\inpSchema,\auxSchema)$,  where $\inpSchema$ and $\auxSchema$ are the schemas of the input database and the auxiliary database, respectively. %
We call relations over $\inpSchema$ \emph{input relations} and relations over $\auxSchema$ \emph{auxiliary relations}. If the relations are $0$-ary, we also speak of input or auxiliary \emph{bits}.
We always let $\schema\df\inpSchema\cup\auxSchema$. 

\begin{definition}(Update program)\label{def:updateprog}
  An \emph{update program} \upProg over a dynamic schema \mbox{$(\inpSchema, \auxSchema)$} 
  is a set of first-order formulas (called \emph{update formulas} in the following) that contains,  for every $R \in \auxSchema$ and every
  $o\in \{\ins_S, \del_S\mid S\in\inpSchema\}$, an update formula  $\uf{R}{o}{\vec x}{\vec y}$ over the schema $\schema$  where $\vec x$ and $\vec y$ have the same arity as $S$ and $R$, respectively.
\end{definition}

A \emph{program state} $\state$ over dynamic schema \mbox{$(\inpSchema, \auxSchema)$} is a structure $(\domain, \inp,  \aux)$ where\footnote{We prefer the notation $(\domain, \inp,  \aux)$ over $(\domain, \inp \cup \aux)$ to emphasize the two components of the overall database.} $\domain$ is a finite domain, $\inp$ is a database over the input schema (the \emph{input database}) and $\aux$ is a database over the auxiliary schema (the \emph{auxiliary database}). %

The \emph{semantics of update programs} is as follows. For a modification $\delta=(o,\vec a)$, where $\vec a$ is a tuple over $\domain$,  and program state $\state=(\domain, \inp,\aux)$ we denote by $\updateState[\upProg]{\delta}{\state}$ the state $(\domain, \updateDB{\delta}{\inp}, \aux')$, where $\aux'$ consists of relations \mbox{$R^{\aux'}\df\{\vec b \mid \state \models \uf{R}{o}{\vec a}{\vec b}\}$}. The effect $\updateState[\upProg]{\alpha}{\state}$ of a modification sequence $\alpha = \delta_1 \ldots \delta_N$ to a state $\state$ is the state $\updateState[\upProg]{\delta_N}{\ldots (\updateState[\upProg]{\delta_1}{\state})\ldots}$. 

\begin{definition}(Dynamic program) \label{definition:dynprog}
  A \emph{dynamic program} is a triple $(\upProg,\init,\querys)$, where
  \begin{compactitem}
   \item  $\upProg$ is an update program over some dynamic schema
  \mbox{$(\inpSchema, \auxSchema)$}, 
    \item \init is a mapping that maps $\inpSchema$-databases to $\auxSchema$-databases, and 
    \item $\querys\in\auxSchema$ is a designated \emph{query symbol}.
  \end{compactitem}
\end{definition}

A dynamic program $\prog=(\upProg,\init,\querys)$ \emph{maintains}  a dynamic query  $\dynProb{$\query$}$ if, for every dynamic instance $(\db,\alpha)$, the query result $\query(\updateDB{\alpha}{\db})$ coincides with the query relation $\querys^\state$ in the state \mbox{$\state=\updateState[\upProg]{\alpha}{\state_\init(\db)}$}, where \mbox{$\state_\init(\db) \df (\domain, \db,  \init(\db))$} is the initial state for $\db$. If the query relation $\querys$ is $0$-ary, we often denote this relation as \emph{query bit} $\Acc$ and say that $\prog$ \emph{accepts} $\alpha$ over $\domain$ if $\Acc$ is true in $\updateState[\upProg]{\alpha}{\state_\init(\db)}$.

In the following, we write $\updateStateI{\alpha}{\db}$ instead of $\updateState[\upProg]{\alpha}{\state_\init(\db)}$ and $\updateStateI{\alpha}{\state}$ instead\footnote{The notational difference is tiny here: we refer to the dynamic program instead of the update program.} of $\updateStateI[\upProg]{\alpha}{\state}$ for a given dynamic program $\prog = (\upProg,\init,\querys)$, a modification sequence $\alpha$, an initial database $\db$ and a state $\state$.

\begin{definition}(\DynFO and \DynProp) \label{definition:dync}
  \DynFO is the class of all dynamic queries that can be maintained by dynamic programs with first-order update formulas  and first-order definable initialization mapping when starting from an initially empty input database. $\DynProp$ is the subclass of $\DynFO$, where update formulas are quantifier-free\footnote{We still allow the use of quantifiers for the initialization.}.
\end{definition}

A $\DynFO$-program is a dynamic program with first-order update formulas, likewise a $\DynProp$-program is a dynamic program with quantifier-free update formulas. 
 A $\DynFOIA{\ell}{m}$-program is a $\DynFO$-program over (at most) $\ell$-ary input databases that uses auxiliary relations of arity at most $m$; likewise for $\DynPropIA{\ell}{m}$-programs.\footnote{We do not consider the case $\ell = 0$ where databases are pure sets with a fixed number of bits.}

Due to the undecidability of finite satisfiability of first-order logic, the emptiness problem---the problem we study first---is undecidable even for \DynFO-programs with only a single auxiliary relation (more precisely, with query bit only). Therefore, we restrict our investigations to fragments of \DynFO. 
Also allowing arbitrary initialization mappings immediately yields an undecidable emptiness problem. This is already the case for first-order definable initialization mappings for arbitrary initial databases.
In the literature classes with various restricted and unrestricted initialization  mappings have been studied, \mbox{see \cite{ZeumeS14icdt}} for a discussion.
In this work, in line with \cite{PatnaikI94}, we allow initialization mappings defined by arbitrary first-order formulas, but require that the initial database is empty.
Of course, we could have studied further restrictions on the power of the initialization formulas, but this would have yielded a setting with an additional parameter.

The following example illustrates a technique to maintain lists with quantifier-free dynamic programs, introduced in \cite[Proposition 4.5]{GeladeMS12}, which is used in some of our proofs. The example itself is from \cite{ZeumeS15}.
\begin{example}\label{example:emptylist}
    We provide a $\DynProp$-program $\prog$ for the dynamic variant of the Boolean query \problem{NonEmptySet}, where, for a unary relation $U$ subject to insertions and deletions of elements, one asks whether $U$ is empty. Of course, this query is trivially expressible in first-order logic, but not without quantifiers.   

   The program $\prog$ is over auxiliary schema  $\auxSchema = \{\querys, \First, \Last, \List\}$, where $\querys$ is the query bit (i.e.\ a $0$-ary relation symbol), $\First$ and $\Last$ are unary relation symbols, and $\List$ is a binary relation symbol. The idea of \prog is to maintain a list of all elements currently in $U$. The list structure is stored in the binary relation $\List^\state$.
The first and last element of the list are stored in $\First^\state$ and $\Last^\state$, respectively. We note that the order in which the elements of $U$ are stored in the list depends on the order in which they are inserted into $U$.
\shortOrLong{}{

For a given instance of \problem{NonEmptySet} the initialization mapping initializes the auxiliary relations accordingly.}%
\shortOrLong{We only describe the (more complicated) case of deletions from $U$.

}{

\insertdescr{U}{a}{
    A newly inserted element is attached to the end of the list\footnote{For simplicity we assume that only elements that are not already in $U$ are inserted, the formulas given can be extended easily to the general case. Similar assumptions are made whenever necessary.}. Therefore the $\First$-relation does not change except when the first element is inserted into an empty set $U$. Furthermore, the inserted element is the new last element of the list and has a connection to the former last element. Finally, after inserting an element into $U$, the query result is 'true': 
    \begin{align*}
      \uf{\First}{\ins}{a}{x} &\df (\neg \querys \mand a = x) \mor (\querys \mand \First(x)) \\
      \uf{\Last}{\ins}{a}{x} &\df a = x \\
      \uf{\List}{\ins}{a}{x,y} &\df \List(x,y) \mor (\Last(x) \mand a = y)  \\
      \uf{\querys}{\ins}{a}{} &\df \top. 
    \end{align*}%
  }}%
  \deletedescr{U}{a}{
    How a deleted element $a$ is removed from the list, depends on whether $a$ is the first element of the list, the last element of the list or some other element of the list. The query bit remains 'true', if $a$ was not the first \emph{and} last element of the list.\shortOrLong{\footnote{We omit the (obvious) parts of formulas that deal with spurious deletions.}}{}
    \begin{align*}
      \uf{\First}{\del_U}{a}{x} &\df (\First(x) \mand x \neq a) \mor (\First(a) \mand \List(a,x)) \\
      \uf{\Last}{\del_U}{a}{x} &\df (\Last(x) \mand x \neq a)  \mor (\Last(a) \mand \List(x,a)) \\
      \uf{\List}{\del_U}{a}{x,y} &\df x \neq a \mand y \neq a \mand \big(\List(x,y) \mor (\List(x, a) \mand \List(a, y))\big)\\      
      \uf{\querys}{\del_U}{a}{} &\df \neg(\First(a) \wedge \Last(a))  
    \end{align*}
\vspace{-10mm}\\
 \qed
  }
\end{example}

In some parts of the paper we will use specific forms of modification sequences. An \emph{insertion sequence} is a modification sequence $\alpha = \delta_1\cdots\delta_m$ whose modifications are pairwise distinct insertions.
An insertion sequence $\alpha$ over a unary input schema $\inpSchema$ is in \emph{normal form} if it fulfills the following two conditions.
\begin{enumerate}[label=(N\arabic*)]
\item For each element $a$, the insertions affecting $a$ form a contiguous subsequence $\alpha_a$ of $\alpha$. We say that $\alpha_a$ \emph{colors} $a$.
\item For all elements $a,b$ that get assigned the same $\inpSchema$-color by $\alpha$, the projections of  the subsequences $\alpha_a$ and $\alpha_b$ to their operations (i.e., their first parameters) are identical.
\end{enumerate}

  \section{The Emptiness Problem}\label{section:emptiness}
    \toAppendix{\section{Proofs for Section \ref{section:emptiness}}}  
\makeatletter{}%
In this section we define and study the decidability of the  emptiness problem for dynamic programs in general and for restricted classes of dynamic programs. The emptiness problem asks, whether the query relation $\querys$ of a given dynamic program $\prog$ is always empty, more precisely, whether $\querys^\state=\emptyset$ for every (empty) initial database $\db$ and every modification sequence $\alpha$ with $\state = \updateStateI[\prog]{\alpha}{\db}$.

To enable a fine-grained analysis, we parameterize the emptiness problem by a class $\calC$ of dynamic programs.

\problemdescr{\Emptiness[$\calC$]}{A dynamic program $\prog\in\calC$ with $\FO$ initialization}{
Is $\querys^\state=\emptyset$, for every initially empty database $\db$ and every modification sequence $\alpha$, where $\state \df \updateStateI[\prog]{\alpha}{\db}$?}

As mentioned before, undecidability of the emptiness problem for unrestricted dynamic programs follows immediately from the undecidability of finite satisfiability of first-order logic. %

\atheorem{theorem:generalundecidability}{
$\Emptiness$ is undecidable for $\DynFOIA{2}{0}$-programs.
}
\toLongAndAppendix{
\begin{proof}
  This follows easily from the undecidability of the finite satisfiability problem for first-order logic over schemas with at least one binary relation symbol \cite{Trahtenbrot63}.
 For a given first-order formula $\varphi$ over schema $\{E\}$ we construct a $\DynFO$-program $\prog$ with a single binary input relation $E$ and a single $0$-ary auxiliary relation $\Acc$ as follows. The bit $\Acc$ is set to true whenever the modified database is a model of $\varphi$, and set to false otherwise.
  
  For correctness, we observe that if $\varphi$ is not satisfiable then $\Acc$ is always false and therefore $\prog$ is empty. On the other hand, if $\varphi$ is satisfiable, then there is a modification sequence $\alpha$ that is accepted by $\prog$, so $\prog$ is non-empty. 
\end{proof}
}

In the remainder of this section, we will shed some light on the border line between decidable and undecidable fragments of \DynFO. In Subsection \ref{section:emptinessgeneral} we study fragments of \DynFO obtained by disallowing quantification and/or restricting the arity of input and auxiliary relations. In Subsection \ref{section:emptinessconsistent}, we consider dynamic programs that come with a certain consistency guarantee.

    \subsection{Emptiness of general dynamic programs}\label{section:emptinessgeneral}
    \toAppendix{\subsection{Proofs for Section \ref{section:emptinessgeneral}}}  
\makeatletter{}%
In this subsection we study the emptiness problem for various restricted classes of dynamic programs. 
We will see that the problem is basically only decidable if all relations are at most unary and no quantification in update formulas is allowed. 
Figure \ref{figure:emptiness:general} summarizes the results.

\begin{figure}[t!]
\centering
\begin{tikzpicture}
\tikzset{
    every node/.style={
        font=\scriptsize
    },
    decision/.style={
        shape=rectangle,
        minimum height=1cm,
        text width=2cm,
        text centered,
        rounded corners=1ex,
        draw
    },
    outcome/.style={
        shape=ellipse,
        fill=gray!15,
        draw,
        text width=2cm,
        text centered
    },
    decision tree/.style={
        edge from parent path={[-latex] (\tikzparentnode) -| (\tikzchildnode)},
        sibling distance=5cm,
        level distance=1.125cm
    },
    level 3/.style={sibling distance=5cm},
    level 4/.style={sibling distance=5cm},
}

\node [decision,
        label={[yshift=0.125cm]left:allowed},
        label={[yshift=0.125cm]right:not allowed}] { Quantification in update formulas }
    [decision tree]
    child { node [outcome] {undecidable (Thm. \ref{theorem:emptiness:undecidables} (a))} }%
    child { node [decision,
        label={[yshift=0.125cm]left:unary},
        label={[yshift=0.125cm]right:binary or more}] { Arity of input relations} 
        child { node [decision,
        label={[yshift=0.125cm]left:at most unary},
        label={[yshift=0.125cm]right:binary or more}] 
        { Arity of auxiliary relations } 
            child { node [outcome] { decidable (Thm. \ref{theorem:emptiness:unarydynprop}) } }
            child { node [outcome] { undecidable (Thm. \ref{theorem:emptiness:undecidables} (b))} }%
        }
        child { node [outcome] { undecidable (Thm. \ref{theorem:emptiness:undecidables} (c))} }%
    };

\end{tikzpicture}
\caption{Decidability of \Emptiness for various classes of dynamic programs.}\label{figure:emptiness:general}
\end{figure}
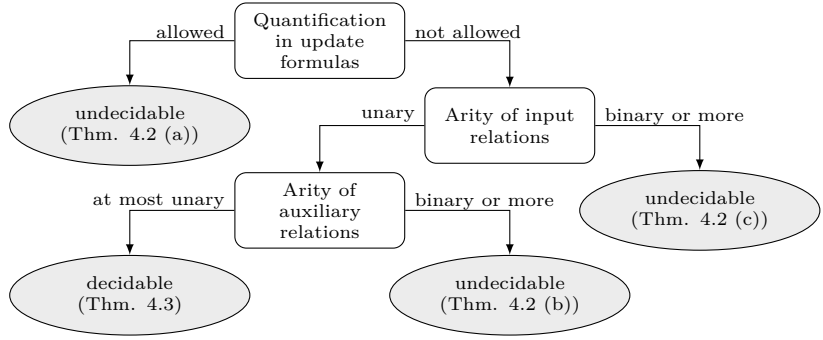

At first we strengthen the general result from Theorem \ref{theorem:generalundecidability}. We show that undecidability of the emptiness problem for  \DynFO-programs
 holds even for unary input relations and auxiliary bits. Furthermore, quantification is not needed to yield undecidability: for \DynProp-programs, emptiness is undecidable for binary input or auxiliary relations.

\atheorem{theorem:emptiness:undecidables}{
  The emptiness problem is undecidable for
  \begin{enumerate}[ref={\thetheorem\ (\alph*)}]
   \item $\DynFOIA{1}{0}$-programs, %
   \item $\DynPropIA{1}{2}$-programs, %
   \item $\DynPropIA{2}{0}$-programs, %
  \end{enumerate}
}
\aproofsketch{}
{
In all three cases, the proof is by a  reduction from the emptiness problem for semi-deterministic 2-counter automata.

In a nutshell, a counter automaton (short: CA) is a finite automaton that is equipped with counters that range over the non-negative integer numbers. A counter $c$ can be incremented ($\inc(c)$), decremented ($\dec(c)$) and tested for zero ($\ifz(c)$). A CA does not read any input (i.e., its transitions can be considered to be $\epsilon$-transitions) and in each step it can manipulate or test one counter and transit from one state to another state. More formally, a CA is tuple $(Q, C,\Delta, q_i , F )$, where $Q$ is a set of states, $q_i \in Q$ is the initial state, $F \subseteq Q$ is the set of accepting states, and $C$ is a finite set (the \emph{counters}). The transition relation $\Delta$ is a subset of 
$Q \times \{\inc(c), \dec(c), \ifz(c) \mid c \in C\} \times Q$. 
A \emph{configuration} of a CA is a pair $(p, \vec n)$ where $p$ is a state and $\vec n \in \N^C$ gives a value $n_c$ for each counter $c$ in $C$. A transition $(p,\inc(c),q)$ can be applied in state $p$, transits to state $q$ and increments $n_c$ by one. A transition $(p,\dec(c),q)$ can be applied in state $p$ if $n_c>0$, transits to state $q$ and decrements $n_c$ by one. A transition $(p,\ifz(c),q)$ can be applied in state $p$,  if $n_c=0$ and transits to state $q$.

A CA is \emph{semi-deterministic} if from every state there is either at most one transition or there are two transitions, one decrementing and one testing the same counter for zero. The emptiness problem for (semi-deterministic) 2-counter automata (2CA) asks whether a given counter automaton with two counters has an accepting run and is undecidable \cite[Theorem 14.1-1]{Minsky1967}.

In all three reductions, the dynamic program \prog is constructed such that for every run~$\rho$ of a semi-deterministic 2CA $\calM$ there is a modification sequence $\alpha=\alpha(\rho)$ that lets \prog simulate~$\rho$,  and an empty  database \db, such that \prog accepts $\alpha$ over \db if and only if $\rho$ is accepting. More precisely, the states of \prog encode the states of $\calM$ by auxiliary bits and the counters of $\calM$ in some way that differs in the three cases.  However, in all cases it holds that not every modification sequence for \prog corresponds to a run of $\calM$. However, \prog can detect if $\alpha$ does \emph{not} correspond to a run and assume a rejecting sink state as soon as this happens.     

For (a), the two counters are simply represented by two unary relations, such that the number of elements in a relation is the current value of the counter. The test whether a counter has value zero thus boils down to testing emptiness of a set and can easily be expressed by a formula with quantifiers. 

The lack of quantifiers makes the reductions for (b) and (c) a bit more complicated. In both cases, the counters are represented by linked lists, similar to Example~\ref{example:emptylist}, and the number of elements in the list corresponds to the counter value (in (c): plus 1). With such a list a counter value zero can be detected without quantification. Due to the allowed relation types, the lists are built with auxiliary relations in (b) and with input relations in (c). 
}
{
\apprepetition{In all three cases, the proof is by a  reduction from the emptiness problem for semi-deterministic 2-counter automata.

In a nutshell, a counter automaton (short: CA) is a finite automaton that is equipped with counters that range over the non-negative integer numbers. A counter $c$ can be incremented ($\inc(c)$), decremented ($\dec(c)$) and tested for zero ($\ifz(c)$). A CA does not read any input (i.e., its transitions can be considered to be $\epsilon$-transitions) and in each step it can manipulate or test one counter and transit from one state to another state.

More formally, a CA is tuple $(Q, C,\Delta, q_i , F )$, where $Q$ is a set of states, $q_i \in Q$ is the initial state, $F \subseteq Q$ is the set of accepting states, and $C$ is a finite set (the \emph{counters}). The transition relation $\Delta$ is a subset of 
$Q \times \{\inc(c), \dec(c), \ifz(c) \mid c \in C\} \times Q$.

A \emph{configuration} of a CA is a pair $(p, \vec n)$ where $p$ is a state and $\vec n \in \N^C$ gives a value $n_c$ for each counter $c$ in $C$. A transition $(p,\inc(c),q)$ can be applied in state $p$, transits to state $q$ and increments $n_c$ by one. A transition $(p,\dec(c),q)$ can be applied in state $p$ if $n_c>0$, transits to state $q$ and decrements $n_c$ by one. A transition $(p,\ifz(c),q)$ can be applied in state $p$,  if $n_c=0$ and transits to state $q$.

A \emph{run} is a sequence of configurations consistent with $\Delta$, starting from the \emph{initial configuration}  $(q_i, \vec 0)$. A run is \emph{accepting}, if it ends in some configuration $(q_f, \vec n)$ with $q_f \in F$.
A CA is \shortOrLong{}{\emph{deterministic} if $\Delta$ contains for every $p \in Q$ at most one transition $(p, \theta, q)$. It is }\emph{semi-deterministic} if for every $p \in Q$ there is at most one transition $(p, \theta, q)$ in $\Delta$ or there are two transitions $(p, \dec(c), q)$ and $(p, \ifz(c), q')$.

The emptiness problem for counter automata asks whether a given counter automaton has an accepting run. 
It follows from \cite[Theorem 14.1-1]{Minsky1967} that the emptiness problem for semi-deterministic CA \emph{with two counters} (2CA) is undecidable.\footnote{The  instruction set from  \cite{Minsky1967} contains the increment instruction and a combined instruction that decrements a counter if it is non-zero and jumps to another instruction if it is zero. To simulate the latter instruction, we use two transitions $(p, \dec(c), q)$ and $(p, \ifz(c), q')$ of which exactly one can be applied.}

In all three reductions, the dynamic program \prog is constructed such that for every run $\rho$ of the 2CA $\calM$ there is a modification sequence $\alpha=\alpha(\rho)$ that lets \prog simulate~$\rho$, and such that \prog accepts on input $\alpha$ if and only if $\rho$ is accepting. More precisely, the state of \prog encodes the state of $\calM$ by auxiliary bits and the counters of $\calM$ in some way that differs in the three cases.  However, in all cases it holds that not every modification sequence for \prog corresponds to a run of $\calM$. However, \prog can detect if $\alpha$ does \emph{not} correspond to a run and assume a rejecting sink state as soon as this happens.     

For (a), the two counters are simply represented by two unary relations, such that the number of elements in a relation is the current value of the counter. The test whether a counter has value zero thus boils down to testing emptiness of a set and can easily be expressed by a formula with quantifiers. 

The lack of quantifiers makes the reductions for (b) and (c) a bit more complicated. In both cases, the counters are represented by linked lists, where the number of elements in the list corresponds to the counter value (in (c): plus 1). With such a list a counter value zero can be detected without quantification. Due to the allowed relation types, the lists are built with auxiliary relations in (b) and with input relations in (c).}

In the following, we describe more details of the reductions.
\begin{proofenum}
\item 
We construct, from a semi-deterministic 2CA
  $\calM = (Q, \{c_1, c_2\},\Delta, q_I , F )$ a Boolean
  $\DynFOIA{1}{0}$-program $\prog$ with unary input relations $C_1$
  and $C_2$ and input bits $Z_1$ and $Z_2$ such that $\calM$ accepts a
  sequence $\theta$ of operations if and only if $\prog$ accepts a
  corresponding sequence $\alpha$ of modifications.

With a run $\rho$  of $\calM$ we can associate an input sequence $\alpha(\rho)$ on a sufficiently large domain as follows: each transition of the form $(p,\inc(c_i),q)$ gives rise to an insertion $\ins_{C_i}(d)$, for some domain value $d$ currently not in $C_i$. Likewise, each operation  $(p,\dec(c_i),q)$ corresponds to a deletion $\del_{C_i}(d)$. Finally, operations  $(p,\ifz(c_i),q)$ correspond alternatingly to operations $\ins_{Z_i}()$ and $\del_{Z_i}()$. 

The semi-determinism of $\calM$ ensures that there is always at most one applicable transition and enables the program $\calP$ to keep track of the state of $\calM$. The program ensures that only applicable transitions are taken.

The program $\prog$ has one auxiliary bit $R_p$ for every state $p$ of $\calM$, an ``error  bit'' $R_e$ and the query bit $\Acc$.
During a ``simulation'' the current state $p$ of $\calM$ corresponds to a program state in which exactly the auxiliary bit $R_p$ is true (and $\Acc$ if $p\in F$). As soon as the input sequence contains an operation that does not correspond to an applicable transition of $\calM$ (either because no transition exists or because it can not be applied due to a counter value), the error bit $R_e$ is switched on and remains on forever.

  The update formulas of $\prog$ are as follows. 
  
  \begin{align*}
      \uf{R_q}{\ins\;C_i}{u}{} &\df \neg C_i(u) \land \neg R_e \land \bigvee_{(p, \inc(c_i) , q) \in \Delta} R_{p} \\ 
      \uf{R_e}{\ins\;C_i}{u}{} &\df R_e \lor C_i(u) \lor  \bigvee_{p\in X}  R_{p}\\ 
      \uf{\Acc}{\ins\;C_i}{u}{} &\df \neg C_i(u) \land \neg R_e \land \bigvee_{\substack{(p, \inc(c_i) , q) \in \Delta \\ \text{with  $q \in F$} }} R_{p}
  \end{align*}
Here, $X$ is the set of states $p$ from $\calM$ for which no transition $(p,\inc(c_i),q)$ exists in $\Delta$.

  Deletions are handled similarly:
  \begin{align*}
      \uf{R_q}{\del\;C_i}{u}{} &\df C_i(u) \land \neg R_e \land \bigvee_{(p, \dec(c_i) , q) \in \Delta} R_{p} \\ 
      \uf{R_e}{\del\;C_i}{u}{} &\df R_e \lor \neg C_i(u) \lor  \bigvee_{p\in Y}  R_{p}\\ 
      \uf{\Acc}{\del\;C_i}{u}{} &\df C_i(u) \land \neg R_e \land \bigvee_{\substack{(p, \dec(c_i) , q) \in \Delta \\ \text{with  $q \in F$} }} R_{p}
  \end{align*}
Here, $Y$ is the set of states $p$ from $\calM$ for which no transition $(p,\dec(c_i),q)$ exists in $\Delta$.
   Modifications to $Z_i$ are handled as follows:
  \begin{align*}
      \uf{R_q}{\ins\;Z_i}{}{} \df &\neg \exists x C_i(x) \land \neg R_e \land \bigvee_{(p, \ifz(c_i) , q) \in \Delta} R_{p} \\ 
      \uf{R_e}{\ins\;Z_i}{}{} \df & R_e \lor \exists x C_i(x) \lor \bigvee_{p\in Z}  R_{p}\\ 
      \uf{\Acc}{\ins\;Z_i}{}{} \df & \neg \exists x C_i(x) \land \neg R_e \land \bigvee_{\substack{(p, \ifz(c_i) , q) \in \Delta \\ \text{with  $q \in F$} }} R_{p}\\
  \end{align*}
Here, $Z$ is the set of states $p$ from $\calM$ for which no transition $(p,\ifz(c_i),q)$ exists in $\Delta$.
  Deletions of input bits are handled exactly like insertions.
  
  Now we prove that $\calM$ has an accepting run if and only if there is a modification sequence accepted by $\prog$. 

(only-if) Let $\rho$ be an accepting run of $\calM$ and let $m$ be the maximum value that a counter of $\calM$ assumes in $\rho$. It is not hard to prove by induction that there is a  modification sequence on every domain with at least $m$ elements that corresponds to $\rho$ in the sense described above.

(if) For the other direction assume that $\alpha = \delta_1 \cdots \delta_n$ is a modification sequence over domain $\domain$ that is accepted by $\prog$. Let $\state_0$ be the initial state of $\prog$ for $\domain$ and let $\state_i$ for $i \in \{1, \ldots, n\}$ be the state reached by $\prog$ after application of $\delta_1 \cdots \delta_i$. Then, by definition of the update formulas of $\prog$ and because $\state_n$ is accepting, the bit $R^{\calS_i}_e$ is not true for any $\state_i$ and no element is inserted into $C_i$ when it was already contained in $C_i$, likewise elements are not deleted from $C_i$ when they are not contained.
The corresponding accepting run of $\calM$ is defined by the sequence $(q_0, \theta_0, q_1) \ldots (q_{n-1}, \theta_{n-1}, q_n)$ of transitions where $q_i$ is the unique state $q$ for which $R^{\state_i}_{q}$ is true. Further the value for $\theta_i$ is $\inc(c_j)$ if $\delta_{i+1}$ inserts an element into $C_j$, $\dec(c_j)$ if $\delta_{i+1}$ deletes an element from $C_j$ and $\ifzero(c_j)$ if $\delta_{i+1}$ modifies~$Z_j$.

\item 
We note that in the proof of part (a) quantification is only needed for testing whether the input relations representing the counters are empty. 

A $\DynPropIA{1}{2}$-program can simulate this check with two lists as in Example~\ref{example:emptylist} for the relations $C_1$ and $C_2$. 
When an insertion $\ins_{C_i}(d)$ occurs, corresponding to an operation $(p, \inc(c_i), q)$ in $\calM$, the element $d$ is appended to the end of the list for $C_i$.
Analogously, for a deletion $\del_{C_i}(d)$ the element $d$ is removed from the list for $C_i$.
As shown in Example~\ref{example:emptylist} the dynamic program maintains auxiliary bits $B_1,B_2$ such that $B_i$ is true if and only if $C_i$ is not empty. These bits can then be used by the update formulas instead of the quantification.
The rest of the proof is then analogous to the proof of (a).

\item 
In this reduction the counters of the CA are represented by lists, as in (b), but the lists are encoded with (at most) binary \emph{input relations}. Consequently, transitions of $\calM$ correspond to (bounded length) \emph{sequences} of modifications for a dynamic program.

For each counter $C_i$ the program \prog use one binary input relation $\List_i$, one unary input relation $\In_i$ that contains all element used in the list, three unary input relations $\RelName{Min}_i$, $\Last_i$, $\RelName{NextLast}_i$ to mark special elements, several auxiliary bits to monitor if all these input relations are used as intended and a bit $\RelName{NonEmpty}_i$ which states whether $\List_i$ is currently empty.

We now describe how to construct a modification sequence $\alpha=\alpha(\rho)$ from a run $\rho$ of a given 2CA $\calM$, that is accepted by \prog if and only if $\rho$ is accepting.

Before the actual simulation of $\calM$ can start, $\alpha$ has to initialize the input relations apart from~$\List_i$. To this end, \prog expects as the first three modifications the insertion of one element into $\RelName{Min}_i, \Last_i$ and $\In_i$. This element will serve as the head of the list.

A transition of $\calM$ that increments counter $c_i$ is translated into a series of modifications that altogether insert a new element $a$ into $\In_i$  as follows. First, $a$ is inserted into $\RelName{NextLast}_i$ and thus marked as to be inserted to the end of the list. Next the tuple $(b,a)$ is inserted into $\List_i$, where $b$ is the unique element with $b \in \Last_i$. The list is surely not empty after the insertion of $a$, so $\RelName{NonEmpty}_i$ is set to true.
After that, $b$ is removed from $\Last_i$ and $a$ is inserted into $\In, \Last_i$ and removed from $\RelName{NextLast}_i$. If the modification sequence does not follow this protocol, \prog assumes a rejecting state forever. Because every relation from $\RelName{Min}_i$, $\Last_i$, $\RelName{NextLast}_i$ contains at most one element at every time, \prog can indeed check whether all these modifications occur in the right order and on the right elements. 

Similarly, a transition of $\calM$ decrementing $c_i$ is translated into a series of modifications that altogether remove the unique element  $a \in \Last_i$ from the corresponding list as follows. Let $(b,a)$ be the tuple in $\List_i$ that contains $a$. The first modification has to be the insertion of $b$ into $\RelName{NextLast}_i$, after that $(b,a)$ is deleted from $\List_i$. If $b \in \RelName{Min}_i$ then the list is now empty and $\RelName{NonEmpty}_i$ is set to false. $a$ has to be removed from $\In$ and $\Last$, $b$ has to be inserted into $\Last$ and removed from $\RelName{NextLast}$.

It is straightforward but cumbersome to give the update formulas, so they are omitted here.

Otherwise, that is, besides the actual translation of a single step of $\calM$, the proof is analogous to the proof of (a).
\end{proofenum}
}

The next result shows that emptiness of $\DynPropIA{1}{1}$-programs is decidable, yielding a clean boundary between decidable and undecidable fragments.

 \atheorem{theorem:emptiness:unarydynprop}{
     \Emptiness is decidable for $\DynPropIA{1}{1}$-programs.
   }
\aproof{
The proof uses the following two simple observations about $\DynPropIA{1}{1}$-programs  \prog.

   \begin{itemize}
      \item The initialization formulas of \prog assign the same $\auxSchema$-color to all elements. This color and the initial auxiliary bits only depend on the size of the domain. Furthermore there is a number $n(\prog)$, depending solely on the initialization formulas, such that the initial auxiliary bits and $\auxSchema$-colors are the same for all empty databases with at least $n(\prog)$ elements. This observation actually also holds for $\DynFOIA{1}{1}$-programs.
      \item When \prog reacts to a modification $\delta=(o,a)$, the new ($\schema$-)color of an element $b\not=a$ only depends on $o$, the old color of $b$, the old color of $a$, and the 0-ary relations. In particular, if two elements $b_1,b_2$ (different from $a$) have the same color before the update, they both have the same new color after the update.  Thus, the overall update basically consists of assigning new colors to each color (for all elements except~$a$), and the appropriate handling of the element~$a$ and the 0-ary relations.
   \end{itemize}

We will show below that the behavior of $\DynPropIA{1}{1}$-programs can be simulated by an automaton model with a decidable emptiness problem, which we introduce next.

A \emph{multicounter automaton} (short: MCA) is a counter automaton which is not allowed to test whether a counter is zero, i.e. the transition relation $\Delta$ is a subset of $Q \times \{\inc(c), \dec(c) \mid c \in C\} \times Q$. 
A \emph{transfer multicounter automaton} (short: TMCA) is a multicounter counter automaton which has, in addition to the increment and the decrement operation, an operation that simultaneously transfers the content of each counter to another counter. More precisely the transition relation $\Delta$ is a subset of $Q \times (\{\inc(c), \dec(c) \mid c \in C\} \cup \{t \mid t: C \rightarrow C\}) \times Q$. Applying a transition $(p, t, q)$ to a configuration $(p, \vec n)$ yields a configuration $(q, \vec n')$ with $n'_c \df \sum_{t(d) = c} n_d$ for every $c \in C$. A configuration  $(q, \vec n)$ of a TCMA is \emph{accepting}, if $q\in F$. The emptiness problem for TCMAs\footnote{We note that (the complement of) this emptiness problem is often called \emph{control-state reachability} problem.} is decidable by reduction to the coverability problem for transfer petri nets\footnote{The simulation of states by counters can be done as in \cite[Lemma 2.1]{HopcroftP79}} which is known to be decidable \cite{DufourdFS98}.

   Let $\prog$ be a $\DynProp$-program over unary schema $\schema = \inpSchema \cup \auxSchema$ with query symbol $\querys$ which may be $0$-ary or unary. Let $\Gamma_0$ be the set of all $0$-ary (atomic) types over $\schema$ and let $\Gamma_1$ be the set of $\schema$-colors. 
   We construct a transfer multicounter automaton $\calM$ with counter set $Z_1 = \{z_{\gamma} \mid \gamma \in \Gamma_1\}$.
   The state set $Q$ of $\calM$ contains $\Gamma_0$, the only accepting state $f$ and some further ``intermediate'' states to be specified below. 

  The intuition is that whenever $\prog$ can reach a state $\state$  then $\calM$ can reach a configuration $c = (p, \vec n)$ such that $p$ reflects the $0$-ary relations in $\state$ and, for every $\gamma\in\Gamma_1$, $n_{\gamma}$ is the number of elements of color $\gamma$  in $\state$. %

  The automaton $\calM$ works in two phases. First, $\calM$ guesses the size $n$ of the domain of the initial database. To this end,  it increments the counter $z_{\gamma}$ to $n$, where $\gamma$ is the color assigned to all elements by the initialization formula for domains of size $n$, and it assumes the state corresponding to the initial 0-ary relations for a database of size $n$. Here the first of the above observations is used. Then $\calM$ simulates an actual computation of $\prog$ from the initial database of size $n$ as follows. Every modification $\ins_S(a)$ (or $\del_S(a)$, respectively) in $\prog$ is simulated by a sequence of three transitions in $\calM$:
   \begin{itemize}
     \item First, the counter $z_{\gamma}$, where $\gamma$ is the color of $a$ before the modification, is decremented.
     \item Second, the counters for all colors are adapted according to the update formulas of $\prog$.
     \item Third, the counter $z_{\gamma'}$, where $\gamma'$ is the color of $a$ after the modification, is incremented.
   \end{itemize}
  If a modification changes an input bit, the first and third step are omitted. The state of $\calM$ is changed to reflect the changes of the 0-ary relations of \prog. For this second phase the second of the above observations is used.

To detect when the simulation of $\prog$ reaches a state with non-empty query relation $\querys$, states $p \in \Gamma_0$ may have a transition to the accepting state $f$.
}{
~ %
}{

Now we describe $\calM$ in detail. We begin with the simulation of the initialization step. If the quantifier depth of $\prog$ is $q$ then $\calM$ non-deterministically guesses whether the domain is of size $1, \ldots, q$ or at least $q+1$. To this end the automaton has $q+1$ additional states $p_{1},\ldots,p_{q+1}$, and non-deterministically chooses one such state $p_i$. Recall that the initial $\auxSchema$-colors as well as the auxiliary bits depend only on the size of the domain, and that they are the same for all domains of size $\geq q+1$. Let $\gamma_0$ be the $0$-ary type and $\gamma_1$ be the color assigned to domains of size $i$. Now, $\calM$ increments the counter $z_{\gamma_1}$ to $i$ (or to at least $i$ if $i = q+1$) using some further intermediate states. Afterwards $\calM$ assumes state $\gamma_0$.

Next we explain how a computation of $\prog$ is simulated. We first deal with modifications to unary input relations. 
As the effects of an update depend on the operation that is applied to an element, the color of that element and the $0$-ary relations, $\calM$ has one chain of transitions for every such combination. So, for every state $p\in \Gamma_0$, every color $\gamma \in \Gamma_1$ and every $o \in \{\ins_S, \del_S\}$ with $S \in \inpSchema$ and $\arity(S)=1$ there are states $q_{p,\gamma,o}^1$ and $q_{p,\gamma,o}^2$ which are in charge of the simulation of an update when the modification $\delta = (o,a)$ occurs in a situation with $0$-ary type $p$ to an element $a$ of color $\gamma$.
 A transition from $p$ to $q_{p,\gamma, o}^1$ decreases the counter $z_{\gamma}$, a transition from $q_{p,\gamma,o}^2$ increases the counter for the new color of the modified element and assumes the state $p'$ corresponding to the new $0$-ary type. These two transitions simulate the changes of the auxiliary relations regarding the modified element.
A transition from $q_{p,\gamma,o}^1$ to $q_{p,\gamma,o}^2$ handles the changes to the elements not (directly) affected by~$\delta$. As explained above, for given $p$, $o$ and $\gamma$, the new color of an element depends only on its old color.
From the update formulas of $\prog$ we extract a function $g_{p,\gamma,o} : \Gamma_1 \rightarrow \Gamma_1$ which describes these changes. From $g$ we build the function $t: Z_1 \rightarrow Z_1$ that describes the transfer as $t(z_{\gamma'}) = z{_{g_{p,\gamma,o} (\gamma')}}$.

Similarly, modifications to input bits are simulated. Let $o \in \{\ins_S, \del_S\}$ with $S \in \inpSchema$ and $\arity(S)=0$ be an operation to a $0$-ary input relation. For states $p, p' \in \Gamma_0$ there is a transition $(p, t, p')$ if $t(z_{\gamma'}) = z{_{g_{p,\gamma,o} (\gamma')}}$ with $g_{p,\gamma,o} : \Gamma_1 \rightarrow \Gamma_1$ as above and $p'$ corresponds to the $0$-ary type after the update.

At last, transitions from $p \in \Gamma_0$ to $f$ are introduced. The kind of these transitions depends on the arity of $\querys$.
If $\querys$ is $0$-ary and $\querys \in p$, then there is a transfer transition $(p, id, f)$ where $id$ is the identity.
If $\querys$ is unary there is a transition $(p,\dec(\gamma),f)$ for every color $\gamma \in \Gamma_1$ with $\querys \in \gamma$.

It is not hard to show that there is a modification sequence for $\prog$ that leads to a non-empty query relation, if and only if there is a run of $\calM$ that reaches~$f$.
}

  \subsection{Emptiness of consistent dynamic programs}\label{section:emptinessconsistent}
    \toAppendix{\subsection{Proofs for Section \ref{section:emptinessconsistent}}}  

\makeatletter{}%

Some readers of the proof of Theorem \ref{theorem:emptiness:undecidables} might have got the impression that we were cheating a bit, since the dynamic programs it constructs do not behave as one would expect: in all three cases each modification sequence $\alpha$ that yields a non-empty query relation $\querys$ can be changed, e.g., by switching two operations, into a sequence that does not correspond to a run of the CA and therefore does \emph{not} yield a non-empty query relation. That is, the program \prog is \emph{inconsistent} because it might yield different results when the same database is reached through two different modification sequences.

It seems, that this inconsistency made the proof of Theorem~\ref{theorem:emptiness:undecidables} much easier. Therefore, the question arises, whether the emptiness problem becomes easier if it can be taken for granted that the given dynamic program is actually consistent. We study this question in this subsection and will investigate the related decision problem whether a given dynamic program is consistent in the next section.

As Table \ref{tab:results} shows, the emptiness problem for consistent dynamic programs is indeed easier in the sense that it is decidable for a considerably larger class of dynamic programs. While emptiness for general \DynFO programs is already undecidable for the tiny fragment with unary input relations and $0$-ary auxiliary relations, it is decidable for consistent \DynFO programs with unary input and unary auxiliary relations. Likewise, for \DynProp there is a significant gap: for consistent programs it is decidable for arbitrary input arities (with unary auxiliary relations) or arbitrary auxiliary arities (with unary input relations), but for general programs emptiness becomes undecidable as soon as binary relations are available (in the input \emph{or} in the auxiliary database).

We call a dynamic program $\prog$ \emph{consistent}, if it maintains a query with respect to an empty initial database, that is, if, for all modification sequences $\alpha$ to an  empty initial database $\db_\emptyset$, the query relation in $\updateStateI{\alpha}{\db_\emptyset}$ depends only on the database $\updateDB{\alpha}{\db_\emptyset}$. 
In the remainder of this subsection we show the undecidability and decidability results stated in  Table \ref{tab:results}. 

\atheorem{theorem:emptiness:consistentfobinary}{
  The emptiness problem is undecidable for
  \begin{enumerate}
\item consistent $\DynFOIA{2}{0}$-programs, and
\item consistent $\DynFOIA{1}{2}$-programs.
  \end{enumerate}
}
\aproofidea{
Statement (a) is a corollary of the proof of Theorem \ref{theorem:generalundecidability}, as the reduction in that proof always yields a consistent program.

For (b), we present another reduction from the emptiness problem for semi-deterministic 2CAs (see also the proof of Theorem \ref{theorem:emptiness:undecidables}). 
From a semi-deterministic 2CA $\calM$ we will construct a consistent Boolean dynamic program $\prog$ with a single unary input relation $U$. The query maintained by \prog is ``$\calM$ halts after at most $|U|$ steps''. Clearly, such a program has a non-empty query result for some database and some modification sequence if and only if $\calM$ has an accepting run.

The general idea is that $\prog$ simulates one step of the run of $\calM$ whenever a new element is inserted to $U$. A slight complication arises from deletions from $U$, since it is not clear how one could simulate $\calM$ one step ``backwards''. Therefore, when an element is deleted from $U$, \prog freezes the simulation and stores the size $m$ of $|U|$ before the deletion. It continues the simulation as soon as the current size $\ell$ of $U$ grows larger than $m$, for the first time. 
}{
~%
}{

\newcommand{\Uc}{\ensuremath{U_{\text{current}}}\xspace}
\newcommand{\Ua}{\ensuremath{U_{\text{acc}}}\xspace}
\newcommand{\Um}{\ensuremath{U_{\text{max}}}\xspace}

To help storing $m$ and $\ell$ (and the values of the counters, for that matter), \prog uses an auxiliary binary relation $R_<$ which, at any time, is a linear order on the set of those elements, that have been inserted to $U$ at some point. Whenever an element is inserted to $U$ for the first time, it becomes the maximum element of the linear order in $R_<$. Deletions and reinsertions do not affect $R_<$. 

To actually store $\ell$ and $m$, \prog uses two unary relations $\Uc$ and $\Um$. At any time, $\Uc$ contains the $\ell$ smallest elements with respect to $R_<$, where $\ell$ is the size of $U$ at the time. Similarly, $\Um$ contains the  $m$ smallest elements, with $m$ as described above. In particular, $\Uc$ is empty if and only if $\ell=0$. 
In the same fashion, \prog uses two further unary auxiliary relations $C_1$ and $C_2$ representing the counters.

If $\calM$ reaches an accepting state, \prog stores the current size $k$ of $U$ at this moment, with the help of another unary relation $\Ua$, that is, it simply lets $\Ua$ become a copy of $\Uc$ after the current insertion. From that point on, that is, if $\Ua$ is non-empty, the query bit of \prog is true whenever $\ell\ge k$. Besides the one binary and five unary relations, \prog has one $0$-ary relation $Q_p$, for every state $p$ of $\calM$.

As an illustration we give two update formulas of \prog that maintain $C_1$ and and $Q_q$, for some state $q$, under insertions to $U$, respectively.

\allowdisplaybreaks
\begin{align*}
      \uf{C_1}{\ins\;U}{u}{x} &\df \big((U(u) \lor  (\Uc\neq\Um) \lor \bigvee_{\substack{(p, \inc(c_2), q) \in \Delta \\ (p, \dec(c_2), q) \in \Delta \\(p, \ifz(c_2), q) \in \Delta}} Q_p) \land C_1(x)\big) \lor \\
        & \Big(\neg U(u) \land (\Uc=\Um) \land \\
      	 & \quad \big(\bigvee_{(p, \inc(c_1), q) \in \Delta} \big( Q_{p} \land \forall y (C_1(y) \lor x \leq y) \big)\\
      	 & \quad \vee \bigvee_{(p, \dec(c_1), q) \in \Delta} \big( Q_p \land C_1(x) \land \exists y (C_1(y) \land x < y) \big)\big)\Big) \\
      \uf{Q_q}{\ins\;U}{u}{} &\df \big((U(u) \lor (\Uc\neq\Um)) \land Q_q \big) \lor \Big( \neg U(u) \land (\Uc=\Um) \land \\
      	 & \big(\bigvee_{\substack{ (p, \inc(c_j), q) \in \Delta\\j\in\{1,2\}}} Q_{p}\\
 &\vee \bigvee_{\substack{ (p, \dec(c_j), q) \in \Delta\\j\in\{1,2\}}} (Q_p \land \exists x C_j(x)) \\
      	 & \vee \bigvee_{\substack{ (p, \ifz(c_j), q) \in \Delta\\j\in\{1,2\}}} (Q_p \land \neg \exists x C_j(x)) \big)\Big)
  \end{align*}

Here,   $\Uc=\Um$ abbreviates the formula $\forall y\; (\Uc(y) \leftrightarrow \Um(y))$. %
We note that $\ufwa{C_1}{\ins\;U}$ does not test the applicability of transitions directly, but $\ufwa{Q_q}{\ins\;U}$ does. 

We recall that, thanks to semi-determinism of $\calM$, the next transition is always uniquely determined by the state of $\calM$ and the value of the affected counter.
If no transition can be applied, the simulation does not set any bit $Q_i$ to true and the simulation basically stops.
}

\shortOrLong{Contrary to the case of not necessarily consistent programs, the emptiness problem is decidable for consistent $\DynFOIA{1}{1}$-programs.
}{}%
\toLongAndAppendix{
Contrary to the case of not necessarily consistent programs, the emptiness problem is decidable for consistent $\DynFOIA{1}{1}$-programs.
We will use the fact that the truth of first-order formulas with quantifier depth $k$ in a state of a $\DynFOIA{1}{1}$-program only depends on the number of elements of every color up to $k$.

Intuitively the states of a consistent $\DynFOIA{1}{1}$-program can be approximated by a finite amount of information, namely the number of elements of every color up to some constant. This can be used to construct, from a consistent $\DynFOIA{1}{1}$-program $\prog$, a nondeterministic finite automaton $\calA$ that reads encoded modification sequences for $\prog$ in normal form and approximates the state of $\prog$ in its own state. In this way the emptiness problem for consistent $\DynFOIA{1}{1}$-programs reduces to the emptiness problem for nondeterministic finite automata.

To formalize this, for a $\DynFOIA{1}{1}$-program $\prog$ let $c_1, \ldots, c_M$ be the colors over the schema of $\prog$. 
The \emph{characteristic vector} $\vec n(\state) = (n_1, \ldots, n_M)$ for a state $\state$ over the schema of $\prog$ stores for every color $c_i$ the number $n_i \in \N$ of elements of color $c_i$ in $\state$. We also denote this number as $n_i(\state)$.
We write $n\simeq_k m$, for numbers $k,n,m$, if $n=m$ or both $n\ge k$ and $m\ge k$.
We write $(n_1,\ldots,n_M) \simeq_k (n'_1,\ldots,n'_M)$, if for every $i \leq M$, $n_i \simeq_k n_i'$, and $\state \simeq_k \state'$ for two states $\state$ and $\state'$ if $\vec n(\state) \simeq_k \vec n(\state')$ and the bits in $\state$ and $\state'$ are equally valuated.

\begin{lemma}\label{lem:PropCons11}
Let $\prog$ be a $\DynFOIA{1}{1}$-program with quantifier depth $q$ and let $\state$ and $\state'$ be two states for $\prog$.
  \begin{enumerate}[ref={\thetheorem\ (\alph*)}]
 \item  $\state \simeq_k \state'$ if and only if $\state \equiv_k \state'$ for any $k \in \N$.\label{lem:CharacVectorimpliesFOtype}
 \item Let $a$ and $a'$ be elements from $\state$ and $\state'$ with the same color $c_i$ and let $k = q+1$. 
If $\state \simeq_k \state'$ and $n_0(\state) \simeq_{k+1} n_0(\state')$ then $\updateState{\delta(a)}{\state} \simeq_k \updateState{\delta(a')}{\state'}$ for every modification $\delta$. \label{lem:UnaryStateProgression}
\end{enumerate}
\end{lemma}

We recall that $\state\equiv_k \state'$ means that the two states satisfy exactly the same first-order formulas of quantifier depth (up to) $k$.

\begin{proof}
\begin{proofenum}
\item It is easy to express with a first-order formula of quantifier depth $k$ that the number of elements of a color $c$ is exactly $k'$ for $k' < k$ or at least $k$.
So the only if direction follows. 
If $\state \simeq_k \state'$ holds, then Duplicator has a straightforward winning strategy in the $k$-rounds Ehrenfeucht-Fra\"isse game, so $\state \equiv_k \state'$ follows.
\item With part (a), $(\state,a) \equiv_k (\state',a')$. Since $k = q+1$, if elements $b$ and $b'$ from $\state$ and $\state'$ have the same color and $b = a$ if and only if $b'=a'$, they also have the same color in $\updateState{\delta(a)}{\state}$ and $\updateState{\delta(a')}{\state'}$.
The claim of the lemma follows.
 \end{proofenum}
\end{proof}

With the help of the previous lemma, we can now show the following decidability result.}

\atheorem{theorem:emptiness:consistentunaryaux}{
  \Emptiness is decidable for consistent $\DynFOIA{1}{1}$-programs. 
}
\aproofidea{}{
The proof uses the fact that the truth of first-order formulas with quantifier depth $k$ in a state of a $\DynFOIA{1}{1}$-program only depends on the number of elements of every color up to $k$. The states of a consistent $\DynFOIA{1}{1}$-program can therefore be abstracted by a bounded amount of information, namely the number of elements of every color up to $k+1$. This can be used to construct, from a consistent $\DynFOIA{1}{1}$-program $\prog$, a nondeterministic finite automaton $\calA$ that reads encoded modification sequences for $\prog$ in normal form and represents the abstracted state of $\prog$ in its own state. In this way the emptiness problem for consistent $\DynFOIA{1}{1}$-programs reduces to the emptiness problem for nondeterministic finite automata.
}{
We reduce the emptiness problem for consistent $\DynFOIA{1}{1}$-programs to the emptiness problem for nondeterministic finite automata.
The intuition is as follows. From a consistent $\DynFOIA{1}{1}$-program $\prog$, we construct a nondeterministic finite automaton $\calA$ that reads encoded modification sequences for $\prog$ in normal form and approximates the state of $\prog$ in its own state.
To this end $\calA$ has a state $q_\calE$ for every equivalence class $\calE$ of $\simeq_k$ for a well-chosen $k \in \N$. The automaton accepts if it reaches a state $q_\calE$ where $\calE$ corresponds to states of $\prog$ with non-empty query relation. 

We make this more precise now. The following facts are exploited in the proof: 
\begin{itemize}
\item As $\prog$ is consistent, if there is a modification sequence that leads to a state with a non-empty query relation, then there is an insertion sequence in normal form that leads to such a state.
\item If two elements $a, a'$ have the same color in some state of the program, then they still have the same color after an element $b \neq a, a'$ has been modified.
\item For knowing how a state $\state$ is updated by $\prog$, it is enough to consider the $\simeq_k$ equivalence class of $\state$ for a suitable $k$.
\end{itemize}

In an insertion sequence in normal form, an element is touched by at most $\ell$ insertions where $\ell$ is the number of unary relation symbols in $\inpSchema$. As the insertions involving a single element occur consecutively in such a sequence, the occurring updates can be specified by ``extended'' update formulas of quantified depth $\ell{}q$, by nesting the original update formulas of quantifier depth $q$. For $k \df \ell{}q + 1$, states $\state$ and $\state'$ with $\state \simeq_k \state'$ then meet the requirements of Lemma \ref{lem:UnaryStateProgression} when those extended update formulas are considered.

The alphabet $\Sigma$ of $\calA$ is the set of \emph{proper} $\inpSchema$-colors ($\neq c_0$). For every equivalence class $\calE$ of $\simeq_k$, for $k$ as chosen above, the automaton $\calA$ has a state $q_{\calE}$. The idea is that the automaton simulates $\prog$ by approximating the state of $\prog$ by its $\simeq_k$-equivalence class. More precisely, whenever $\calA$ is in state $q_\calE$ after reading a word $w$ over $\Sigma$ then $\calE$ is the equivalence class of the state $\state$ reached by $\prog$ after the modification sequence $\alpha$ corresponding to $w$.

There is a small caveat to this. The state reached by $\prog$ after application of $\alpha$ is not solely determined by $\alpha$ but also by the size of the domain. The automaton has to take this into account. 

We now describe the behaviour $\calA$ in detail. At the beginning of a computation the automaton non-deterministically guesses the (approximate) size of the domain, that is, a number $i$ from $\{1,\ldots, k\}$ and assumes state $q_\calE$ where $\calE$ is the equivalence class of $\simeq_k$ that corresponds to an initial state of $\prog$ with $i$ elements if $i < k$ and at least $i$ elements otherwise. Note that if $i = k$ then the automaton does not know the exact size of the domain for which it is simulating $\prog$. Yet, as long as there are at least $k$ $\inpSchema$-uncolored elements, the exact number is not important.

Afterwards $\calA$ simulates $\prog$. When in state $q_\calE$ and reading a symbol $c$, the automaton assumes state $q_{\calE'}$ where $\calE'$ is as follows:

\begin{itemize}
  \item If $\calE$ indicates less than $k$ $\inpSchema$-uncolored elements then $\calE'$ is the $\simeq_k$-equivalence class of any state $\state'$ reached by $\prog$ from a state $\state$ with $\simeq_k$-equivalence class $\calE$. 
  \item	If $\calE$ indicates at least $k$ $\inpSchema$-uncolored elements, then $\calA$ guesses whether this is still the case after coloring one further element. If yes, then $\calE'$ is the $\simeq_k$-equivalence class of any state $\state'$ reached by $\prog$ from a state $\state$ with $\simeq_k$-equivalence class $\calE$ and at least $k+1$ $\inpSchema$-uncolored elements. Otherwise $\calE'$ is the $\simeq_k$-equivalence class of any state $\state'$ reached by $\prog$ from a state $\state$ with $\simeq_k$-equivalence class $\calE$ and at least $k$ $\inpSchema$-uncolored elements. 
\end{itemize}

That $\calE'$ is uniquely determined follows from the second and third fact from above.

}

\makeatletter{}%
The picture of decidability of emptiness for consistent programs for all classes of the form $\DynFOIA{\ell}{m}$ is pretty clear and simple: it is decidable if and only if $\ell= 1$ \emph{and} $m\le 1$. Now  we turn our focus to the corresponding classes of  consistent $\DynProp$-programs. Here we do not have a full picture. We show in the following that it is decidable if $\ell= 1$ \emph{or} $m\le 1$.
\begin{theorem}\label{theorem:emptiness:consistent:unaryDynProp}
  The emptiness problem is decidable for 
  \begin{enumerate}
   \item consistent  $\DynPropI{1}$-programs. 
   \item consistent $\DynPropA{1}$-programs. 

  \end{enumerate}

\end{theorem}
\shortOrLong{
\begin{proofideaof}{Theorem \ref{theorem:emptiness:consistent:unaryDynProp} (a)}
The statement follows almost immediately from the fact that every consistent $\DynPropI{1}$-program with $0$-ary query relations maintains a regular language \cite[Theorem 3.2]{GeladeMS12}.
\end{proofideaof}
}{
\begin{proofof}{Theorem \ref{theorem:emptiness:consistent:unaryDynProp} (a)}
In \cite[Theorem 3.2]{GeladeMS12} it is shown that over databases with a linear order and unary relations every $\DynPropI{1}$-program~$\prog$ with a Boolean query relation maintains a regular language over the $\inpSchema$-colors of the $\inpSchema$-colored elements.  This result holds for arbitrary initialization and its proof shows that an automaton for this regular language can be effectively constructed from the dynamic program. Therefore, to test emptiness of a program with a Boolean query relation it suffices to test emptiness of its automaton. 

Suppose that $\prog$ has a query relation with arity $k > 0$ and that there is a modification sequence $\alpha$ that yields a state $\state$ where the query relation contains a tuple $\vec a \df (a_1, \ldots, a_k)$. Without loss of generality we assume that $\alpha$ is an insertion sequence in normal form and that elements of $\vec a$ are modified at last (if they are modified at all). In other words, $\alpha$ is of the form $\alpha_1 \ldots \alpha_M$ where each $\alpha_i$ modifies exactly one element, and there is an $N$ such that $\alpha_j$ with $j \geq N$ only modifies elements of $\vec a$.

We use a pumping argument to argue that if $\alpha$ is a shortest such sequence, then it is not very long. Then emptiness of $\prog$ can be tested by examining all such modification sequences. We use the following observations from \cite[Theorem 3.2]{GeladeMS12}:
\begin{enumerate}
 \item After each update, all tuples of positions that have not been touched so far have the same (atomic) type.
 \item There is only a bounded number (depending only on the number and the maximal arity of the auxiliary relations of $\prog$) of different types of such tuples.
\end{enumerate}
Let $\calS_i$ be the state reached by applying $\alpha_1 \ldots \alpha_i$. If $N$ is larger than the number of (atomic) $k$-ary types then, by the observations (a) and (b), there are $j$, $j'$ with $j < j'$ such that all $l$-tuples whose elements have not been touched so far have the same type in $\state_j$ and~$\state_{j'}$. In particular $\vec a$ has the same type in $\state_j$ and $\state_{j'}$. Hence, since $\prog$ is quantifier-free, it also has the same type in $\state$ (the state reached by applying $\alpha$) and in the state reached by applying the modification sequence $\alpha_1 \ldots \alpha_j \alpha_{j'+1} \ldots \alpha_N \alpha_{N+1} \ldots \alpha_M$. Thus the query relation contains $\vec a$ in the latter state.
\end{proofof}
}

\shortOrLong{To highlight the role of the Sunflower Lemma for the proof of Theorem \ref{theorem:emptiness:consistent:unaryDynProp} (b), we give a full exposition of this proof in the following. At first, we sketch the basic proof idea for consistent $\DynPropA{1}$-programs over graphs, i.e., the input schema contains a single binary relation symbol $E$.}
{Before we prove the general statement of Theorem \ref{theorem:emptiness:consistent:unaryDynProp} (b), we first sketch the basic proof idea for consistent $\DynPropA{1}$-programs over graphs, i.e., the input schema contains a single binary relation symbol $E$.}
For simplicity we also assume a $0$-ary query relation. The general statement requires more machinery and is proved below.
  
  Our goal is to show that if such a program $\prog$ accepts some graph then it also accepts one with ``few'' edges, where ``few'' only depends on the schema of the program. To this end we show that if a graph $G$ accepted by $\prog$ contains many edges then one can find a large ``well-behaved'' edge set in $G$ from which edges can be removed without changing the result of~$\prog$. Emptiness can then be tested in a brute-force manner by trying out insertion sequences for all graphs with few edges (over a canonical domain $\{1,\ldots,n\}$).
  
  More concretely, we consider an edge set ``well-behaved'', if it consists only of self-loops, it is a set of disjoint non-self-loop-edges, or is is a \emph{star}, that is, the edges share the same source node or the same target node. From the Sunflower Lemma \cite{ErdosR60} it follows that for every $p \in \N$ there is an $N_p \in \N$ such that every (directed) graph with $N_p$ edges contains $p$ self-loops, or $p$ disjoint edges, or a star with $p$ edges. 

  Let us now assume, towards a contradiction, that the minimal graph accepted by $\prog$ has $N$ edges with $N>N_{M^2+1}$, where $M$ is the number of binary (atomic) types over the schema $\tau=\inpSchema\cup\auxSchema$ of $\prog$. Then $G$ either contains $M^2+1$ self-loops, or $M^2+1$ disjoint edges, or a $(M^2+1)$-star. 
  
  Let us assume first that $G$ has a set $D \subseteq E$ of $M^2+1$ disjoint edges. We consider the state $\state$ reached by $\prog$ after inserting all edges from $E \setminus D$ into the initially empty graph. Since $D$ contains $M^2+1$ edges, there is a subset $D' \subseteq D$ of size $M+1$ such that all edges in $D'$ have the same atomic type in state $\state$. Let $\state_0$ be the state reached by $\prog$ after inserting all edges in $D \setminus D'$ in  $\state$. All edges in $D'$ still have the same type in $\state_0$ since $\prog$ is a quantifier-free program (though this type can differ from the type in $\state$). Let  $e_1, \ldots, e_{M+1}$ be the edges in $D'$ and denote by $\state_i$ the state reached by $\prog$ after inserting  $e_1, \ldots, e_i$ in $\state_0$. For each $i$, all edges $e_{i+1},\ldots, e_{M+1}$ have the same type $\gamma_i$ in  state $\state_i$, again. As the number of binary atomic types is $M$, there are $i<j$ such that $\gamma_i=\gamma_j$, thus $e_{M+1}$ has the same type in $\state_i$ and~$\state_{j}$. Therefore, inserting the edges  $e_{j+1}, \ldots, e_{M+1}$ in $\state_i$ yields a state with the same query bit as inserting those edges in $\state_j$. As the query bit in the latter case is accepting, it is also accepting in the former case, yet in that case the underlying graph has fewer edges than~$G$, the desired contradiction.
  The case where $G$ contains $M^2+1$ self-loops is completely analogous.

  Now assume that $G$ contains a star with $M^2+1$ edges. The argument is very similar to the argument for disjoint edges. First insert all edges not involved in the star into an initially empty graph. Then there is a set $D$ of many  star  edges of the same type, and they still have the same type after inserting the other edges of the star. A graph with fewer edges that is accepted by~$\prog$ can then be obtained as above. 
  
  The idea generalizes to input schemata with larger arity by applying the Sunflower Lemma in order to obtain a ``well-behaved'' sub-relation within an input relation that contains many tuples. 
  In order to prove this generalization we first recall the Sunflower Lemma, and observe that it has an analogon for tuples. 

The Sunflower Lemma was introduced in \cite{ErdosR60}, here we follow the presentation in \cite{Jukna01}. A \emph{sunflower} with $p$ \emph{petals} and a \emph{core} $Y$ is a collection of $p$ sets $S_1, \ldots, S_p$ such that $S_i \cap S_j = Y$ for all $i \neq j$.

\begin{lemma}[Sunflower Lemma, \cite{ErdosR60}]\label{lemma:sunflower}
 Let $p \in \N$ and let $\calF$ be a family of sets each of cardinality $\ell$. If $\calF$ consists of more than $N_{\ell, p} \df \ell!(p-1)^\ell$ sets then $\calF$ contains a sunflower with $p$ petals. 
\end{lemma}

We call a set $H$ of tuples of some arity $\ell$ a \emph{sunflower (of tuples)} if it has the following three properties.
\begin{enumerate}[label=(\roman*)]
\item All tuples in $H$ have the same equality type.
\item There is a set $J\subset\{1,\ldots,\ell\}$ such that $t_j=t'_j$ for every $j\in J$ and all tuples $t,t'\in H$.
\item For all tuples $t\not=t'$ in $H$ the sets $\{t_i\mid i\not\in J\}$ and   $\{t'_i\mid i\not\in J\}$ are disjoint.
\end{enumerate}
We say that $H$ has $|H|$ petals.

The following Sunflower Lemma for tuples has been stated in various variants in the literature, e.g., in \cite{Marx05,KratschW10}.

\begin{lemma}[Sunflower Lemma for tuples]\label{lemma:sunflower-tuples}
Let  $\ell, p \in \N$ and let $R$ be a set of $\ell$-tuples. If $R$ contains more than $\bar{N}_{\ell, p}\df \ell^\ell p^\ell(\ell!)^2$ tuples then it contains a sunflower with $p$ petals.
\end{lemma}

\begin{proof}
  Let $R$ be an $\ell$-ary relation that contains $\bar{N}_{\ell, p}$ tuples. As there are less than $\ell^\ell$ equality types of $\ell$-tuples there is a set $R'\subseteq R$ of size at least $p^\ell(\ell!)^2$, in which all tuples have the same equality type. Application of Lemma 2 in \cite{KratschW10} yields\footnote{In \cite{KratschW10}, elements from the ``outer part'' of a petal can also occur in the ``core''. As in $R'$ all tuples have the same equality type, this can not happen in our setting.} a sunflower with $p$ petals.
\end{proof}
It is instructive to see how Lemma \ref{lemma:sunflower-tuples} shows that a graph with sufficiently many edges has many selfloops, disjoint edges or a large star: Selfloops correspond to the equality type of tuples $(t_1,t_2)$ with $t_1=t_2$, many disjoint edges to the case $J=\emptyset$ and the two possible kinds of stars to $J=\{1\}$ and $J=\{2\}$, respectively.

\begin{proofof}{Theorem \ref{theorem:emptiness:consistent:unaryDynProp} (b)}
  Now the proof for binary input schemas easily translates to general input schemas. For the sake of completeness we give a full proof.
  
  Suppose that a consistent $\DynPropA{1}$-program $\prog$ over schema $\schema$ with $0$-ary\footnote{At the end of the proof we discuss how to deal with unary query relations.} query relation accepts an input database $\db$ that contains at least one relation $R$ with many tuples.

  Suppose that $R$ is of arity $\ell$ and contains $\bar{N}_{\ell, M^2+1}$ diverse tuples where $M$ is the number of $\ell$-ary (atomic) types over the schema of $\prog$. We show that $\prog$ already accepts a database with less tuples than~$\db$.
  
  By Lemma \ref{lemma:sunflower-tuples}, $R$ contains a sunflower $R'$ of size $M^2+1$. Consider the state $\state$ reached by $\prog$ after inserting all tuples from $R \setminus R'$ into the initially empty database. Since $R'$ contains $M^2+1$ tuples, there is a subset $R'' \subseteq R'$ of size $M+1$ such that all tuples in $R''$ have the same atomic type in state $\state$. Let $\state_0$ be the state reached by $\prog$ after inserting all tuples in $R' \setminus R''$ in  $\state$. All tuples in $R''$ still have the same type in $\state_0$ since $\prog$ is a quantifier-free program (though this type can differ from the type in $\state$). 

  Let  $\vec a_1, \ldots, \vec a_{M+1}$ be the tuples in $R''$ and denote by $\state_i$ the state reached by $\prog$ after inserting  $a_1, \ldots, a_i$ in $\state_0$. In state $\state_i$ all tuples $a_{i+1},\ldots, a_{M+1}$ have the same type, again. As the number of $\ell$-ary atomic types is $k$, there are $i<j$ such that $a_{M+1}$ has the same type in $\state_i$ and $\state_{j}$. Therefore, inserting the edges  $e_{j+1}, \ldots, e_{M+1}$ in $\state_i$ yields a state with the same query bit as inserting this sequence in $\state_j$. As the query bit in the latter case is accepting, it is also accepting in the former case, yet in that case the underlying database has fewer tuples than~$\db$, the desired contradiction.

  If $\prog$ has a unary query relation, then the proof has to be adapted as follows. For an accepted database $\db$, the unary query relation contains some element $a$. Now $M$ is chosen as the number of $(\ell+1)$-ary atomic types (instead of the number of $\ell$-ary atomic types), and $R''$ is chosen as sub-sunflower where all tuples $(\vec a_1, a), \ldots, (\vec a_{M+1}, a)$ have the same atomic type. The rest of the proof is analogous.
\end{proofof}

\toMainAndAppendix{
  The final result of this subsection gives a characterization of the class of queries maintainable by consistent $\DynPropA{0}$-programs. This characterization is not needed to obtain decidability of the emptiness problem for such queries, since this is included in Theorem \ref{theorem:emptiness:consistent:unaryDynProp}. However, we consider it interesting in its own right.
}

\toLongAndAppendix{
 As $\DynPropA{0}$-programs can only store a constant amount of information, it is not surprising that they can only maintain very simple properties. In fact, they can maintain exactly all modulo-like queries (to be defined precisely below). This characterization immediately yields an alternative emptiness test for  consistent $\DynPropA{0}$-programs. Furthermore it partially answers a question by Dong and Su \cite{DongS97}. They asked whether all queries maintainable by $\DynFOA{0}$-programs can already be maintained by history-independent $\DynFOA{0}$-programs. The characterization shows that this is the case for $\DynProp$-programs, since all modulo-like queries can easily be maintained by history-independent $\DynPropA{0}$-programs.

\newcommand{\modexpr}[3][\gamma]{\ensuremath{\#(#1)\equiv_{#2} #3}}
\newcommand{\ktype}[1][k]{\ensuremath{#1\text{-type}}}
\newcommand{\Stype}{\ensuremath{\text{type}}}
\newcommand{\dom}{\ensuremath{\text{dom}}}
\newcommand{\stup}{\ensuremath{\text{st}}}

We first fix some notation. For a tuple $\vec a = (a_1,\ldots,a_k)$ we write $\dom(\vec a)$ for the set $\{a_1,\ldots,a_k\}$.
The \emph{cardinality} of $\vec a$ is the size of $\dom(\vec a)$.
 The \emph{strict underlying tuple} $\stup(\vec a)$ is the tuple obtained from $\vec a$ by removing all duplicate occurrences of data values (in a left-to-right fashion). 
A tuple $\vec a$ is \emph{duplicate-free} if $\stup(\vec a)=\vec a$. 

A  \emph{strict atomic $k$-atom} is a relation atom $R(y_1,\ldots,y_r)$ for which $\{y_1,\ldots,y_r\}=\{x_1,\ldots,x_k\}$ with $x_i \neq x_j$ for $i \neq j$. 
 A \emph{strict atomic $k$-type} $\gamma(x_1,\ldots,x_k)$ is a set of strict atomic $k$-atoms.
Let, for a tuple $\vec a = (a_1,\ldots,a_k)$,
$\iota$ be the valuation that maps, for each $j\in\{1,\ldots,k\}$, $x_j$ to $a_j$.  
 Then the \emph{strict atomic type} $\gamma$ of tuple $\vec a = (a_1,\ldots,a_k)$ in \state is the set of strict atomic $k$-atoms $R(y_1,\ldots,y_r)$ in $\gamma$, for which $\iota(R(y_1,\ldots,y_r))$ yields a fact in \state. We write $\ktype(\vec a)$ for the strict atomic type of a $k$-tuple $\vec a$.

However, the expressive power of consistent $\DynPropA{0}$-programs can be most easily characterized in terms of types of sets of elements, rather than types of tuples.

The \emph{set type} $\Stype(A)$ of a set $A=\{a_1,\ldots,a_k\}$ of size $k$ in  a structure \state is the set $\{\ktype(\pi(\vec a))\mid \pi\in S_k\}$. Here, $S_k$ denotes the set of permutations on $\{1,\ldots,k\}$ and $\pi(\vec a)$ denotes the tuple $(a_{\pi(1)},\ldots,a_{\pi(k)})$. We note that $\Stype(A)$ does not depend on the chosen enumeration of $A$ and is therefore well-defined. 
It directly follows from this definition that the set types of two sets with $k$ elements are either equal or disjoint (as sets of strict atomic $k$-types). In other words, the strict atomic type of a set is determined by the strict atomic $k$-type of each duplicate-free tuple that can be constructed from elements of the set. %

 For a structure \state and a set type $\gamma$, we denote by $\#_\state(\gamma)$ the number of sets of set type $\gamma$ in $\state$. 

A \emph{simple modulo expression} is an expression of the form $\modexpr{p}{q}$, where $p\ge 2$ and $q<p$ are natural numbers and $\gamma$ is a non-empty set type. A structure \state satisfies such an expression if $\#_\state(\gamma)\equiv_p q$, that is, if the number of sets of type $\gamma$ in \state has remainder $q$ when divided by $p$. A \emph{modulo expression} is a Boolean combination of simple modulo expressions. A \emph{modulo query} is a query that can be defined as the set of all (finite) models of some modulo expression. 

In the proof of the following theorem, we will consider modification sequences of a particular form that extends the normal form for insertion sequences over unary input schemas introduced in Section \ref{section:setting}. 
A general insertion sequence $\alpha$ is in \emph{normal form} if it fulfills the following three conditions.
\begin{enumerate}[label=(M\arabic*)]
\item If $\alpha$ inserts tuples of cardinality $k$ over a set $A$ of $k$ elements, then all such tuples are inserted in a contiguous subsequence $\alpha_A$ of $\alpha$. Furthermore if $\alpha_A$ and $\alpha_{A'}$ are the contiguous sequences for sets $A$ and $A'$  with $|A| > |A'|$ then $\alpha_A$ occurs before $\alpha_{A'}$ in $\alpha$.
\item For all sets $A,B$ with the same set type in \inp, the subsequences $\alpha_A$ and $\alpha_B$ are isomorphic, that is, for some bijection $\pi:A\to B$, $\pi(\alpha_A)=\alpha_B$.
\end{enumerate}

}
\atheorem{theorem:characterization}{
 A Boolean query $\query$ can be maintained by a consistent $\DynPropA{0}$ program if and only if it is a modulo query.
}
\shortOrLong{Intuitively\footnote{The actual formalization uses sets of elements rather than tuples.}, a \emph{modulo query} is a Boolean combination of constraints of the form: the number of tuples that have some atomic type $\gamma$ is $q$ modulo $p$, for some natural numbers $p\ge 2$ and~$q<p$.}{}

\toLongAndAppendix{
\begin{proof} (if) The set of Boolean queries that can be expressed by consistent  $\DynPropA{0}$ programs is closed under all Boolean operators. It therefore suffices to show that each query defined by a simple modulo expression $\modexpr{p}{q}$ can be maintained by a  consistent $\DynPropA{0}$ program \prog. 

The insertion of a tuple $\vec b$ into some relation $R$ changes the set type of exactly one set, $\{b_1,\ldots,b_r\} \df \dom(\vec b)$. It is straightforward but tedious to construct a quantifier-free formula  $\varphi^R_\gamma(y_1,\ldots,y_r)$ that expresses that the new type of the set $\{b_1,\ldots,b_r\}$ after inserting $\vec b$ to $R$ is $\gamma$. Likewise, for the old set type of $\{b_1,\ldots,b_r\}$. For deletions the situation is very similar. 
A $\DynPropA{0}$ program can therefore use $p$ auxiliary bits to maintain the number of occurrences of set type $\gamma$ in \state modulo $p$.

(only-if) Let $\prog$ be a consistent $\DynPropA{0}$-program. As \prog is consistent it yields, for each input database \inp,  the same query answer, for each modification sequence that results in \inp. In this proof we therefore only consider insertion sequences in normal form. 

Condition (M2) ensures that when a tuple $\vec b$ is inserted to a relation $R$, there are no tuples present that involve a strict subset of $\dom(\vec b)$. As, on the other hand, due to the lack of quantifiers, the update formulas for the auxiliary bits can not take any tuples into account that contain elements outside of $\dom(\vec b)$,
the auxiliary bits of \prog after an insertion operation $\ins_R(\vec b)$ of $\alpha$ only depend on the current auxiliary bits of \prog and the strict atomic $k$-type of $\stup(\vec b)$. 
Similarly, by Condition (M3) it follows that the auxiliary bits after a modification subsequence $\alpha_A$ only depend on the current auxiliary bits of \prog and the set type of $A$. The behavior of \prog under a  insertion sequence in normal form is therefore basically the behavior of a finite automaton (with the possible values of the auxiliary bits as states) reading a sequence of set types.\footnote{It should be noted here, that the overall number of set types is finite and only depends on the signature of \prog.} 

Let $m$ be the number of ($0$-ary) auxiliary bits of \prog and let $M=(2^m)!$.  

We next show that, for each non-empty set type $\gamma$ and each two input databases $\inp$ and $\inp'$ that have for each non-empty set type different from $\gamma$ the same number of sets and whose number of sets of type $\gamma$ differs by $M$, either both  $\inp$ and $\inp'$ are accepted by \prog, or both are rejected. As there are only finitely many types and finitely many classes modulo $M$, this yields that the query decided by \prog is a modulo query.

Let $\state=(\domain, \inp,  \aux)$ be some state reached after an insertion sequence $\alpha$ in normal form, let $\gamma$ be some non-empty set type and let $s$ be the number of occurrences of $\gamma$ in \inp. Let $\alpha'$ be the extension of $\alpha$ by $M+2^m$ further sets of type $\gamma$ yielding  $\state'=(\domain, \inp',  \aux')$. 
Let $A_1,\ldots,A_s$ denote the sets of type $\gamma$ in $\inp$ and let $A_1,\ldots,A_{s'}$ denote  the sets of type $\gamma$ in $\inp'$. 
Let $\alpha'$ be decomposed into $\alpha_1\alpha_{A_1}\cdots\alpha_{A_{s'}}\alpha_2$.\footnote{Note that $\alpha$ has the form $\alpha_1\alpha_{A_1}\cdots\alpha_{A_{s}}\alpha_2$.} As there are only $2^m$ different possible values that the auxiliary bits can assume, there are $i<j$, $j\le 2^m$, such that $\alpha_1\alpha_{A_1}\cdots\alpha_{A_i}$ and $\alpha_1\alpha_{A_1}\cdots\alpha_{A_j}$ yield states with identical auxiliary bits.\footnote{Here, $i=0$ corresponds to the sequence $\alpha_1$.} As each set $A_\ell$ has the same set type, it follows that $\alpha_1\alpha_{A_1}\cdots\alpha_{A_{i+cd}}$ yields the same auxiliary bits as  $\alpha_1\alpha_{A_1}\cdots\alpha_{A_i}$, for $d \df j-i$ and every $c$ with $i+cd\le s+M+2^m$.
If $s\ge i$ it follows that $\alpha_1\alpha_{A_1}\cdots\alpha_{A_{i+M}}$ yields the same auxiliary bits as  $\alpha_1\alpha_{A_1}\cdots\alpha_{A_i}$ and that $\alpha_1\alpha_{A_1}\cdots\alpha_{A_{s+M}}$ yields the same auxiliary bits as  $\alpha_1\alpha_{A_1}\cdots\alpha_{A_s}$. Let us now assume that $s<i$. By deleting $i-s$ sets of type $\gamma$ from the state reached after $\alpha_1\alpha_{A_1}\cdots\alpha_{A_i}$ and $\alpha_1\alpha_{A_1}\cdots\alpha_{A_{i+M}}$, we obtain states with identical auxiliary bits and $s$ and $s+M$ sets of type $\gamma$, respectively. The claim then follows by adding back $\alpha_2$ to the sequences $\alpha_1\alpha_{A_1}\cdots\alpha_{A_s}$  and $\alpha_1\alpha_{A_1}\cdots\alpha_{A_{s+M}}$, respectively. This completes the proof.
\end{proof}

}

   \subsection{The impact of built-in orders}\label{section:emptinessbuiltin}
    \toAppendix{\subsection{Proofs for Section \ref{section:emptinessbuiltin}}}  
\makeatletter{}%
A closer inspection of the proof that the emptiness problem is undecidable for consistent $\DynFOIA{1}{2}$-programs (Theorem \ref{theorem:emptiness:consistentfobinary}) reveals that the construction only requires one binary auxiliary relation: a linear order on the ``active'' elements. The proof would also work if a global linear order on all elements of the 
domain would be given. We say that a dynamic program has a \emph{built-in linear order}, if there is one auxiliary relation $R_<$ that is always initialized by a linear order on the domain and never changed. Likewise, for a \emph{built-in successor relation}.

That is, the border of undecidability for consistent $\DynFO$-programs  actually lies between consistent $\DynFOIA{1}{1}$-programs and consistent $\DynFOIA{1}{1}$-programs with a built-in linear order. Similarly, the border of undecidability for (not necessarily consistent) $\DynProp$-programs actually lies between $\DynPropIA{1}{1}$-programs and $\DynPropIA{1}{1}$-programs with a built-in linear order.

\aproposition{prop:builtin-undec}{
The emptiness problem is undecidable for 
  \begin{enumerate}[ref={\thetheorem\ (\alph*)}]
   \item consistent $\DynFOIA{1}{1}$-programs with a built-in linear order or a built-in successor relation,
  \item $\DynPropIA{1}{1}$-programs with a  built-in successor relation.
 \end{enumerate}
 }
\aproof{}{}{
  \begin{proofenum}
  \item The only binary auxiliary relation used in the
    proof of Theorem \ref{theorem:emptiness:consistentfobinary} was to
    simulate a linear order on the domain. This is not necessary any
    more, if the linear order is available. The linear order can easily be replaced by a  built-in successor relation.
 \item We adapt the proof of Theorem
    \ref{theorem:emptiness:undecidables} (b) and use the successor relation
    instead of the list relations, which are the only binary auxiliary
    relations.  The first modification touches an element that is then
    marked as the first and last element of both lists.  We then
    demand that an insertion to $C_i$ inserts the element that is
    marked as last and a deletion from $C_i$ deletes the predecessor
    of the last element. This can be checked and the marking of the
    last element can be updated without the use of quantifiers.  A
    relation $C_i$ is empty after the element that is marked as first
    is deleted from $C_i$.
   \end{proofenum}
}

However, for dynamic programs that only have auxiliary bits, linear
orders or successor relations do not affect decidability.

\aproposition{prop:builtin-dec}{
The emptiness problem is decidable for 
  \begin{enumerate}[ref={\thetheorem\ (\alph*)}]
   \item consistent $\DynFOIA{1}{0}$-programs with a built-in linear order or a built-in successor relation,
  \item $\DynPropIA{1}{0}$-programs with a built-in linear order or a built-in successor relation.
 \end{enumerate}
} 
\aproof{}{}{  
  \begin{proofenum}
  \item 
    Let $\prog$ be a consistent program over unary input relations
    that uses only $0$-ary auxiliary relations and a built-in linear
    order.  In \cite [Theorem 3.1]{DongW97} it is shown\footnote{We note that the setting in that paper assumes a built-in linear order.} how to construct an
    existential monadic second order formula $\varphi$ such that there
    is a modification sequence $\alpha$ with
    $\updateStateI{\alpha}{\emptyDB}$ is accepted by $\prog$ if
    and only if $\updateDB{\alpha}{\emptyDB} \models \varphi$.  By
    \cite{BuechiE58}, the formula $\varphi$ describes a regular
    language over the proper $\inpSchema$-colors ($\neq c_0$).
 Hence an
    equivalent finite state automaton can be constructed.  For finite
    automata the emptiness problem is decidable, so the claim follows.
  \item This statement simply follows from the decidability of the emptiness problem for $\DynPropIA{1}{1}$-programs (Theorem \ref{theorem:emptiness:unarydynprop}) and the fact that the update formulas of $\DynPropIA{1}{0}$-programs only have one variable and therefore can not use a linear order or a successor relation in a non-trivial way. 
  \end{proofenum}
}%

   \section{The Consistency Problem}\label{section:consistency}
    \toAppendix{\section{Proofs for Section \ref{section:consistency}}}  
\makeatletter{}%

In Section \ref{section:emptinessconsistent} we studied \Emptiness for classes of consistent dynamic programs. It turned out that with this restriction the emptiness problem is easier than for general dynamic programs. 
One might thus consider the following approach for testing whether a given general dynamic program is empty: Test whether the program is consistent and if it is, use an algorithm for consistent programs. To understand whether this approach can be helpful, we study the following algorithmic problem, parameterized by a class $\calC$ of dynamic programs. 

\problemdescr{\Consistency[$\calC$]}{A dynamic program $\prog \in \calC$ with $\FO$ initialization}{Is $\prog$ a consistent program with respect to empty initial databases?}

\shortOrLong{}{We will see that the mentioned approach does not give us any advantage, as deciding \Consistency is as hard as deciding \Emptiness for general dynamic programs.}
It is not very surprising that \Consistency is not easier than \Emptiness, since deciding \Emptiness boils down to finding  \emph{one} modification sequence resulting in a state with particular properties and \Consistency is about finding \emph{two} modification sequences resulting in two states with particular properties.
This high level comparison can actually be turned into rather easy reductions from \Emptiness to \Consistency.

On the other hand, \Consistency can also be reduced to \Emptiness. For this direction the key idea is to simulate two modification sequences simultaneously and to integrate their resulting states into one joint state.
This is easy if quantification is available, and requires some work for $\DynProp$-fragments. \toLongAndAppendix{We first give a technical lemma to restrict the kind of modification sequences that have to be considered to decide \Consistency. 

For this, we use the notion of \emph{innocuous transformations}. Intuitively, an innocuous transformation $\theta$ of a modification sequence $\alpha$ is a minimal change of $\alpha$ that results in a modification sequence $\theta(\alpha)$ which leads to the same underlying database as $\alpha$. 
Formally, an innocuous transformation is either (1) a permutation of a subsequence $\delta_1\delta_2$ to $\delta_2\delta_1$ under the condition that if one modification is $\ins_S(\vec a)$ then the other one is not $\del_S(\vec a)$, (2) the removal of a subsequence $\ins_S(\vec a)\del_S(\vec a)$ if $\vec a$ is not contained in $S$ when this subsequence is applied, (3) the removal of a modification $\delta = \ins_S(\vec a)$ or $\delta = \del_S(\vec a)$ if $\vec a$ is already contained in $S$ respectively $\vec a$ is not contained in $S$ when the modification is applied, or (4) the inverse of one of these transformations.
It is clear that under the given conditions, for an innocuous transformation $\theta$ of a modification sequence $\alpha$ it holds that $\updateDB{\alpha}{\emptyDB} = \updateDB{\theta(\alpha)}{\emptyDB}$.

\begin{lemma}\label{lemma:consistency:editdistance}
Let $\prog$ be an inconsistent dynamic program. Then there is a modification sequence $\alpha$, an innocuous transformation $\theta$ of $\alpha$ and an empty database $\emptyDB$ such that the query relations in $\updateStateI{\alpha}{\emptyDB}$ and $\updateStateI{\theta(\alpha)}{\emptyDB}$ differ.
\end{lemma}
\begin{proof}
As $\prog$ is inconsistent, there are two modification sequences $\alpha$ and $\alpha'$ that lead to the same input database $\inp$ but to states with different query relations.
It is easy to see that $\alpha' = \theta_1\cdots\theta_M(\alpha)$ where each $\theta_i$ is an innocuous transformation of $\theta_1\cdots\theta_{i-1}(\alpha)$: From $\alpha$ and $\alpha'$ we can obtain a common insertion sequence $\alpha''$ by applying innocuous transformations of type (1)-(3), the inverses of the latter sequence of transformations then yields $\alpha'$ from $\alpha''$.  
As $\alpha$ and $\alpha'$ lead to states with different query relations there must be an $i$ such that $\alpha^\star \df \theta_1\cdots\theta_{i-1}(\alpha)$ and $\theta_i(\alpha^\star)$ lead to states with different query relations. 
\end{proof}

We now give the reductions between \Consistency and \Emptiness. 
}
\atheorem{theorem:reductions}{
Let $\ell\ge 1, m\ge 0$.
  \begin{enumerate}
    \item For every $\calC\in\{\DynFOIA{\ell}{m},\DynFOI{\ell},\DynFOA{m},\DynFO\}$,\\
    (i) $\Emptiness(\calC) \le \Consistency(\calC)$, and
  (ii)  $\Consistency(\calC) \le \Emptiness(\calC)$.

    \item For every $\calC\in\{\DynPropIA{\ell}{m},\DynPropI{\ell},\DynPropA{m},\DynProp\}$,\\
    (i) $\Emptiness(\calC) \le \Consistency(\calC)$, and
    (ii)  $\Consistency(\calC) \le \Emptiness(\calC)$.
  \end{enumerate}  
}

\toLongAndAppendix{
\begin{proof}
For (a)(i) and (b)(i), we construct dynamic programs whose query relations are inflationary, that is, tuples that are inserted once are never removed afterwards. When an update adds a tuple and the modification that caused that update is undone, the two states that are reached after these updates are witnesses to inconsistency.

For (a)(ii) and (b)(ii), the constructed dynamic programs simulate two independent modification sequences and maintain two states of the original program.
For (a)(ii), the program uses quantification to determine whether the two states represent equal input databases but different query relations. 
For (b)(ii) we use that thanks to Lemma \ref{lemma:consistency:editdistance} it suffices to simulate one modification sequence and at one point one innocuous transformation to find witnesses for inconsistency, so the two maintained states always represent equal input databases.
\begin{proofenum}
  \item[(a)(i)]
  For a given $\DynFOIA{\ell}{m}$-program $\prog$ over schema $(\inpSchema, \auxSchema)$ with query symbol $\querys$ we construct a $\DynFOIA{\ell}{m}$-program $\prog'$ over $(\inpSchema , \auxSchema \cup \{\querys'\})$ with query symbol $\querys'$. 
 The idea is to initialize $\querys'$ as empty and add the tuples in $\querys$ to $\querys'$ with a delay of one modification. No tuple gets removed from $\querys'$.
 The update formulas for $\querys'$ are $\uf{\querys'}{o}{\vec x}{\vec y} \df \querys(\vec y) \vee \querys'(\vec y)$.
 The update formulas for relations from $\auxSchema$ are copied from~$\prog$. 
 
If $\prog$ is empty, then $\querys'^{\state} = \emptyset$ in every reached state $\state$ and $\prog'$ is consistent. If $\prog$ is non-empty, then let $\alpha$ be a shortest modification sequence such that $\querys^{\updateStateI{\alpha}{\emptyDB}}$ is non-empty and let $\alpha^\star = \alpha\alpha'$ be a modification sequence that leads to the same input database as $\alpha$. 
It follows that the query relation $\querys'$ differs in $\updateStateI[\prog']{\alpha}{\emptyDB}$ and $\updateStateI[\prog']{\alpha^\star}{\emptyDB}$ and $\prog'$ is inconsistent.
\item[(a)(ii)] If $\prog$ is a given $\DynFOIA{\ell}{m}$-program, we construct a $\DynFOIA{\ell}{m}$-program $\prog'$ that simulates two modification sequences of $\prog$ in parallel and maintains two states of this program. 
If the input databases of theses states are equal, a tuple is added to the query relation of $\prog'$ if it is included in exactly one of the two maintained query relations of $\prog$. 

If $\prog$ is over schema $(\inpSchema, \auxSchema)$, then $\prog'$ is over schema $(\inpSchema', \auxSchema')$ where $\inpSchema' = \{S, S'\ \mid S \in \inpSchema\}$ and $\auxSchema' = \{R, R'\ \mid R \in \auxSchema\} \cup \{\querys^\star\}$. The query relation of $\prog$ is $\querys^\star$.
The update formulas of relations $R \in \auxSchema$ are the same as in $\prog$, for relations $R' \in \auxSchema$ the update formulas are obtained from the original formulas by replacing every relation symbol $S \in \inpSchema$ or $R \in \auxSchema$ by $S'$ or $R'$, respectively. 
 The update formulas for $\querys^\star$ first check if the two maintained input databases are equal by using conjunctions of formulas $\forall \vec{x} (S(\vec{x}) \leftrightarrow S'(\vec{x}))$ for every $S \in \inpSchema$ and then inserts a tuple if it is in exactly one of the query relation $\querys$ of $\prog$ and its copy $\querys'$. 
 $\prog$ is consistent if and only if $\prog'$ is empty.
\item[(b)(i)] Analogous to (a)(i).
\item[(b)(ii)] We adapt the idea of part (a)(ii) with the help of Lemma \ref{lemma:consistency:editdistance}.
  For a $\DynPropIA{\ell}{m}$-program $\prog$ over schema $(\inpSchema, \auxSchema)$ we sketch the construction of a $\DynPropIA{\ell}{m}$-program $\prog'$ over schema $(\inpSchema', \auxSchema')$. 
 Like in part (a)(ii), this program simulates two modification sequences of $\prog$ and maintains two auxiliary databases over $\auxSchema$, but only one input database over $\inpSchema$. 
 Contrary to (a)(ii), $\prog'$ either simulates the effects of one modification to both auxiliary databases or, exactly once,  a subsequence (of length at most 2) and an innocuous transformation of this subsequence. 
 It follows that the input databases are equal for both simulated modification sequences after every simulated modification and so $\prog'$ only has to check whether there are tuples that are included in exactly one copy of the original query relation.
 
 We now sketch the construction of $\prog'$. Like in part (a)(ii), $\auxSchema'$ contains relation symbols $R, R'$ for every $R \in \auxSchema$. Also all relation symbols from $\inpSchema$ are contained in $\inpSchema'$.
Additionally, $\inpSchema'$ contains relation symbols $U_S$, $I_S$ and $T_{S}, T'_{S}$ for every $S \in \inpSchema$ to simulate subsequences and their innocuous transformations.
$U_S$ is for simulating an unnecessary modification. If a modification $\ins_{U_S}(\vec a)$ is applied to $\prog'$, the update formulas check that $\vec a$ is already contained in $S$. If this check fails, $\prog'$ sets an error bit. Otherwise, $\prog'$ simulates $\prog$ for the modification $\ins_S(\vec a)$ on the second copy of the auxiliary database. Analogously for a modification $\del_{U_S}(\vec a)$.
When a modification $\ins_{I_S}(\vec a)$ occurs, $\prog'$ simulates $\prog$ for the sequence $\ins_S(\vec a)\del_S(\vec a)$ on the second copy, if $\vec a$ is not contained in $S$ before. Otherwise, $\prog'$ sets an error bit.
A sequence $\ins_{T_S}(\vec a)\ins_{T'_{S'}}(\vec b)\del_{T'_{S'}}(\vec b)\del_{T_S}(\vec a)$ is used to simulate the sequence $\ins_S(\vec a)\del_{S'}(\vec b)$ on the first copy of the auxiliary database and the sequence $\del_{S'}(\vec b)\ins_S(\vec a)$ on the second copy, likewise for other combinations of insertions and deletions. Some additional auxiliary bits are used to check that four modifications like this happen in a row and that they do not represent the insertion of a tuple to a relation and the deletion of that tuple from the same relation.
We use additional auxiliary bits to maintain whether exactly one innocuous transformation has been simulated.
For every modification over relation symbols from $\inpSchema$, both copies of the auxiliary database get updated according to the original program $\prog$. 

It follows from Lemma \ref{lemma:consistency:editdistance} that it is possible for $\prog'$ to reach a state where the copies $\querys$ and $\querys'$ of the query relation of $\prog$ differ and no error bit is set if and only if $\prog$ is inconsistent.
A tuple is inserted into the query relation $\querys^\star$ of $\prog'$ when no error bit is set and the tuple is in exactly one of $\querys$ and $\querys'$.
So $\prog'$ is empty if and only if $\prog$ is consistent.
\end{proofenum}
\end{proof}
}

  \section{The History Independence problem}\label{section:hi}
    \toAppendix{\section{Proofs for Section \ref{section:hi}}}  
\makeatletter{}%
\newcommand{\invector}[1]{{\vec n}^\text{in}(#1)}
\newcommand{\invectorcomp}[2]{n^\text{in}_{#2}(#1)}
\newcommand{\invectorterm}{characteristic input vector}
\newcommand{\invectortermplural}{characteristic input vectors}

\shortOrLong{As discussed in Section \ref{section:emptinessconsistent}, it is natural to expect that a dynamic program is consistent, i.e., that the query relation only depends on the current database, but not on the modification sequence by which it has been reached. 
Many dynamic programs in the literature satisfy a stronger property: not only their query relation but \emph{all} their auxiliary relations depend only on the current database. Formally, we call a dynamic program \emph{history independent} if all auxiliary relations in \updateStateI{\alpha}{\db}  only depend on $\alpha(\db)$, for all modification sequences $\alpha$ and initial empty databases~$\db$. 
History independent dynamic programs (also called \emph{memoryless} \cite{PatnaikI94} or \emph{deterministic} \cite{DongS97}) are still expressive enough to maintain interesting queries like undirected reachability \cite{GraedelS12}, but also some lower bounds are known for such programs \cite{DongS97,GraedelS12, ZeumeS15}. 

In this section, we study decidability of the question whether a given dynamic program is history independent.

\problemdescr{\HI[$\calC$]}{A dynamic program $\prog \in \calC$ with $\FO$ initialization}{Is $\prog$ history independent with respect to empty initial databases?}

Not surprisingly, \HI is undecidable in general. This can be shown basically in the same way as the general undecidability of \Emptiness in Theorem \ref{theorem:generalundecidability}. 

\begin{theorem}\label{theorem:hi:generalundecidability}
\HI is undecidable for $\DynFOIA{2}{0}$-programs.
\end{theorem}

However, the precise borders between decidable and undecidable fragments are different for \HI than for \Emptiness and \Emptiness for consistent programs. More precisely, \HI is decidable for $\DynFO$- and \DynProp-programs with unary input databases, and for $\DynProp$-programs with unary auxiliary databases. 

For showing that \HI is decidable for $\DynFO$-programs with unary input databases, we prove that if such a program is not history independent then this is witnessed by some reachable small ``bad state''. A decision algorithm can then simply test whether such a state exists. Bad states satisfy one of two properties: they either locally contradict history independence or they are ``inhomogeneous''. We define both notions in the following. 

A state $\state$ over domain $\domain$ is \emph{\lhi}\footnote{We define this term for arbitrary input arity, since the first part of Lemma \ref{lemma:hicharacterization} holds in general.} for a dynamic program $\prog$ if the following three conditions hold.
\begin{enumerate}[label={(H\arabic*)}]%
          \item $\updateState{\delta_1 \delta_2}{\state} = \updateState{\delta_2 \delta_1}{\state}$ for all insertions $\delta_1$ and $\delta_2$.
          \item $\state = \updateState{\ins_R(\vec a)\del_R(\vec a)}{\state}$ if $\vec a \notin R^{\state}$, for all $R \in \inpSchema$ and $\vec a$ over $\domain$.
          \item $\state = \updateState{\ins_R(\vec a)}{\state}$ if $\vec a \in R^{\state}$ and $\state = \updateState{\del_R(\vec a)}{\state}$ if $\vec a \notin R^{\state}$, for all $R \in \inpSchema$ and $\vec a$ over $\domain$.
        \end{enumerate}

Locally history independence is well-suited to algorithmic analysis. The following lemma shows that for testing history independence it actually suffices to test locally history independence for all states reachable by very restricted modification sequences.

\begin{lemma}\label{lemma:hicharacterization}
  Let $\prog$ be a dynamic program. 
    \begin{enumerate}[ref={\thetheorem\ (\alph*)}]
      \item $\prog$ is history independent if and only if every state reachable by $\prog$ via insertion sequences is \lhi.
      \item  If $\prog$ is a $\DynFOI{1}$-program, then $\prog$ is history independent if and only if every state reachable by $\prog$ via insertion sequences in normal form is \lhi.
    \end{enumerate}
\end{lemma}

A state \state is \emph{homogeneous} if all tuples  $\vec a$ and $\vec b$ with the same (atomic) \inpSchema-type also have the same (atomic) \auxSchema-type.  The following lemma is an immediate consequence of \cite[Lemma 16]{DongS97}.

\begin{lemma}\label{lem:homogeneous}
  For every history independent $\DynFOI{1}$-program, every reachable state is homogeneous. 
\end{lemma}

We call a state of a $\DynFOI{1}$-program that is  not homogeneous or not \lhi a \emph{bad state}. The following lemma is the key ingredient for deciding history independence for $\DynFOI{1}$-programs. It restricts the size of the smallest bad state and therefore allows for testing history independence in a brute-force manner.

\begin{proposition}\label{prop:smallmodel}
   Let \prog be a $\DynFOIA{1}{m}$-program. There is a number $N \in \N$, that can be computed from \prog, such that if \prog is \emph{not} history independent, then there exists an empty database $\emptyDB$ of size at most $N$ and an insertion sequence  $\alpha$ in normal form such that $\updateStateI{\alpha}{\emptyDB}$ is bad.
\end{proposition}

\begin{theorem}\label{theorem:hidecidable}
\HI is decidable for $\DynFOI{1}$-programs.
\end{theorem}

Using the same technique as in the proof of Theorem \ref{theorem:emptiness:consistent:unaryDynProp}(b), history independence can be shown to be decidable for $\DynPropA{1}$-programs.
\begin{theorem}\label{theorem:hidynpropdecidable}
  \HI is decidable for $\DynPropA{1}$-programs.  
\end{theorem}
}
{

As discussed in Section \ref{section:emptinessconsistent}, it is natural to expect that a dynamic program is consistent, i.e., that the query relation only depends on the current database, but not on the modification sequence by which it has been reached. 
Many dynamic programs in the literature satisfy a stronger property: not only their query relation but \emph{all} their auxiliary relations depend only on the current database. Formally, we call a dynamic program \emph{history independent} if all auxiliary relations in \updateStateI{\alpha}{\db}  only depend on $\alpha(\db)$, for all modification sequences $\alpha$ and initial empty databases~$\db$. 
History independent dynamic programs (also called \emph{memoryless} \cite{PatnaikI94} or \emph{deterministic} \cite{DongS97}) are still expressive enough to maintain interesting queries like undirected reachability \cite{GraedelS12}, but also some lower bounds are known for such programs \cite{DongS97,GraedelS12, ZeumeS15}. 

In this section, we study decidability of the question whether a given dynamic program is history independent.

\problemdescr{\HI[$\calC$]}{A dynamic program $\prog \in \calC$ with $\FO$ initialization}{Is $\prog$ history independent with respect to empty initial databases?}

Note that contrary to the emptiness problem, \HI is not easier for classes of consistent dynamic programs than for classes of general dynamic programs, so we will not study this restriction. This is because for every dynamic program $\prog$ we can construct a consistent dynamic program $\prog'$ that is history independent if and only if $\prog$ is, by introducing a new query bit that is not changed by any update formula.

Not surprisingly, \HI is undecidable in general. This can be shown basically in the same way as the general undecidability of \Emptiness in Theorem \ref{theorem:generalundecidability}. 

\begin{theorem}\label{theorem:hi:generalundecidability}
\HI is undecidable for $\DynFOIA{2}{0}$-programs.
\end{theorem}

\begin{proof}
  Again we reduce from the satisfiability problem for first-order logic over schemas with at least one binary relation symbol. 
  For a given $\FO$-formula $\varphi$, at first we construct the dynamic program $\prog$ from the proof of Theorem \ref{theorem:generalundecidability}. Additionally we add a second auxiliary bit $B$ which is initialized as false and set to true when $\Acc$ is first set to true by an update, and never set to false again. 
  If $\varphi$ is not satisfiable, then all bits remain false and $\prog$ is history independent. If $\varphi$ is satisfiable, then let $\alpha\delta$ be a shortest modification sequence applied to an empty database $\emptyDB$ such that $\Acc$ and $B$ are set to true in $\updateStateI{\alpha\delta}{\emptyDB}$.
  Let $\delta^{-1}$ be the modification that undoes $\delta$.
  Then $B$ is false in $\updateStateI{\alpha}{\emptyDB}$ and true in $\updateStateI{\alpha\delta\delta^{-1}}{\emptyDB}$, but the respective input databases are equal. So $\prog$ is not history independent.
\end{proof}

However, in the following we will see that the precise borders between decidable and undecidable fragments are different for \HI than for \Emptiness and \Emptiness for consistent programs.
More precisely, we will show that \HI is decidable for $\DynFO$- and \DynProp-programs with unary input databases, and for $\DynProp$-programs with unary auxiliary databases. %

We recall the normal form for insertion sequences introduced in Section \ref{section:setting}. For dynamic programs with unary input databases, insertion sequences in normal form (1) color each element contiguously and (2) apply the insertions for each $\inpSchema$-color in the same order. Here we require further that they first color all elements with designated $\inpSchema$-color $c_1$, then all elements with $c_2$ and so on. 

We will first show that to judge \HI of $\DynFOI{1}$-program only modification sequences in normal form (Lemma \ref{lemma:hicharacterization}) and states with a particular property (Lemma \ref{lem:homogeneous}) need to be considered. Finally, we show that if a dynamic program is not history independent, this can be observed already for domains of a bounded size in the size of the program (Proposition \ref{prop:smallmodel-app}).
The decision algorithm then tests all states over such ``small'' domains reached by insertion sequences in normal form  in a brute-force manner.  

Let $\prog$ be a $\DynFOI{1}$-program over schema $\schema = \inpSchema \cup \auxSchema$.
Throughout this section we assume that $\schema$ contains only at least unary relation symbols and no input or auxiliary bits to ease presentation. This is no real restriction, as these bits can easily be simulated by unary relations when quantification is allowed. We usually denote the maximum quantifier depth of (initialization and update) formulas by $q$, the maximum arity of aux-relations by $m$,  and the number of input relations by $\ell$. 
Further we write $L$ for $2^{\ell}-1$ and let $c_0, \ldots, c_L$ be the $\inpSchema$-colors, where $c_0$ is the color of the $\inpSchema$-uncolored elements.
In the following we write ``colors'' and ``uncolored'' instead of $\inpSchema$-colors and $\inpSchema$-uncolored.

We next present a characterization of history independence which is well-suited to algorithmic analysis. We call a state $\state$ over domain $\domain$  \emph{\lhi}\footnote{We define this term for arbitrary input arity, since the first part of Lemma \ref{lemma:hicharacterization} holds in general.} for a dynamic program $\prog$ if the following three conditions hold.
\begin{enumerate}[label={(H\arabic*)}]%
          \item $\updateState{\delta_1 \delta_2}{\state} = \updateState{\delta_2 \delta_1}{\state}$ for all insertions $\delta_1$ and $\delta_2$.
          \item $\state = \updateState{\ins_R(\vec a)\del_R(\vec a)}{\state}$ if $\vec a \notin R^{\state}$, for all $R \in \inpSchema$ and $\vec a$ over $\domain$.
          \item $\state = \updateState{\ins_R(\vec a)}{\state}$ if $\vec a \in R^{\state}$ and $\state = \updateState{\del_R(\vec a)}{\state}$ if $\vec a \notin R^{\state}$, for all $R \in \inpSchema$ and $\vec a$ over $\domain$.
        \end{enumerate}

\begin{lemma}\label{lemma:hicharacterization}
  Let $\prog$ be a dynamic program. 
    \begin{enumerate}[ref={\thetheorem\ (\alph*)}]
      \item $\prog$ is history independent if and only if every state reachable by $\prog$ via insertion sequences is \lhi.
      \item  If $\prog$ is a $\DynFOI{1}$-program, then $\prog$ is history independent if and only if every state reachable by $\prog$ via insertion sequences in normal form is \lhi.
    \end{enumerate}
\end{lemma}

\begin{proof}
\begin{proofenum}
\item 
(only-if) It is easy to see that local history independence for all reachable states is necessary for history independence. 

(if) Assume, towards a contradiction, that there is a dynamic program $\prog$, for which every state reachable by an insertion sequence is \lhi, but $\prog$ is not history independent. Then there are two modification sequences $\alpha_1$ and $\alpha_2$ to an empty database $\emptyDB$ with $\updateDB{\alpha_1}{\emptyDB} = \updateDB{\alpha_2}{\emptyDB}$ but $\updateStateI{\alpha_1}{\emptyDB} \neq \updateStateI{\alpha_2}{\emptyDB}$.
We construct insertion sequences $\alpha_1'$ and $\alpha_2'$ that lead to the same state as  $\alpha_1$ and $\alpha_2$, respectively.  Repeated application of (H1) then yields $\updateStateI{\alpha'_1}{\emptyDB} = \updateStateI{\alpha'_2}{\emptyDB}$ and altogether $\updateStateI{\alpha_1}{\emptyDB}=\updateStateI{\alpha'_1}{\emptyDB} = \updateStateI{\alpha'_2}{\emptyDB} = \updateStateI{\alpha_2}{\emptyDB}$, the desired contradiction.

We only describe how to construct the insertion sequence $\alpha'_1$ from $\alpha_1$; the construction of $\alpha'_2$ from $\alpha_2$ is completely analogous. Let thus $\alpha_1 = \delta_1 \cdots \delta_N$ and, for every $i$, we denote by $\state_i\df\updateStateI{\delta_1 \cdots \delta_i}{\emptyDB}$.

A modification is \emph{bad} if it is a deletion or the repeated insertion of a fact. 
The insertion sequence $\alpha'_1$ is constructed by successively eliminating all bad modifications from $\alpha_1$. If $\alpha_1$ does not contain any bad modification, we are done. Otherwise, let $\delta_k$ be the first bad modification in~$\alpha_1$. Since $\delta_1\cdots \delta_{k-1}$ is an insertion sequence, by our assumption $\state_{k-1}$ is \lhi.  Therefore, $\delta_k$ can be eliminated from $\alpha_1$ as follows. If $\delta_k = \del_R(\vec a)$ such that $\vec a \notin R^{\state_{k-1}}$ or $\delta_k = \ins_R(\vec a)$ such that $\vec a \in R^{\state_{k-1}}$ then $\state_k=\state_{k-1} $ thanks to (H3) and $\delta_k$ can be removed from $\alpha_1$ without affecting the resulting state.
If $\delta_k = \del_R(\vec a)$ such that $\vec a \in R^{\state_{k-1}}$, then there must be an insertion $\ins_R(\vec a)$ in~$\delta_1 \cdots \delta_{k-1}$. By (H1) the insertions $\delta_1 \cdots \delta_{k-1}$ can be rearranged into a sequence $\beta \ins_R(\vec a)$, such that $\beta$ consists of all modifications from $\delta_1 \cdots \delta_{k-1}$ besides $\ins_R(\vec a)$ and the resulting state is $\state_{k-1}$.  By (H2), the modification sequences $\beta $ and $\beta \ins_R(\vec a) \del_R(\vec a)$ yield the same state, but $\beta$ has fewer deletions than $\delta_1 \cdots \delta_k$. The modification sequence $\alpha'_1$ is obtained by repeating this procedure.
\item
(only-if) Again, local history independence for all reachable states is necessary for history independence. 

(if) Let $\prog$ be a dynamic program for which every state reachable via a insertion sequence in normal form is \lhi. We show that every state reachable by an insertion sequence is also reachable by a normal form sequence. That \prog is history independent then follows from (a).

We thus assume, towards a contradiction, that there is an insertion sequence $\alpha = \delta_1 \cdots \delta_N$ and an empty database $\emptyDB$ such that $\state=\updateState{\alpha}{\emptyDB}$ is not reachable by any insertion sequence in normal form. Let $\alpha$ and $\emptyDB$ be chosen such that $N$ is minimal. Therefore, $\state'=\updateState{\delta_1\cdots\delta_{N-1}}{\emptyDB}$ can be reached by a normal form modification\footnote{Of course, insertion sequences yielding the same state have the same length.} sequence $\alpha'=\delta'_1\cdots\delta'_{N-1}$ and, by our assumption, $\state'$ and all prior states reached by prefixes of $\alpha'$ are \lhi. By inductive application of (H1), $\delta_N$ can now be moved to its appropriate place inside $\alpha'$ to yield a normal form sequence $\alpha''$ equivalent to $\alpha$. Therefore, $\state$ is reachable by a normal form sequence, the desired contradiction.
\end{proofenum}
\end{proof}

We next define another property that reachable states of history independent programs share. A state \state is \emph{homogeneous} if all tuples  $\vec a$ and $\vec b$ with the same (atomic) \inpSchema-type also have the same (atomic) \auxSchema-type. For every homogeneous state \state we denote by $f_\state$ the \emph{(atomic) type function} that maps every (atomic) \inpSchema-type of arity $m$ (the maximal arity of $\schema$) to the corresponding (atomic) \auxSchema-type\footnote{If there is no tuple $\vec a$ of an $\inpSchema$-type $c$ in $\state$, then $f_\state(c) = \bot$}. The following lemma is an immediate consequence of \cite[Lemma 16]{DongS97}.

\begin{lemma}\label{lem:homogeneous}
  For every history independent $\DynFOI{1}$-program, every reachable state is homogeneous. 
\end{lemma}

We call a state of a $\DynFOI{1}$-program that is  not homogeneous or not \lhi a \emph{bad state}.
That a state is bad can be expressed in first-order logic. Likewise the possible effects of coloring a single uncolored element on the type function of a state can be expressed by first-order formulas. To state this more precisely, we use \emph{type forecast functions} $F:\{1,\ldots,L\}\to \calF$, where $\calF$ is the set of possible type functions for \prog.

\begin{lemma}\label{lem:formulas}
 Let $\prog$ be a $\DynFOIA{1}{m}$-program with maximum quantifier-depth $q$ and $\ell$ input relations.
  \begin{enumerate}
  \item There is a formula $\phibad$ of quantifier-depth at most $3+2m+(\ell+1)q$ that is true in a state $\state$ if and only if $\updateState{\alpha}{\state}$ is bad for at least one modification sequence $\alpha$ that colors   a single uncolored  element of \state.
  \item For every type forecast function $F$ there is a formula $\varphi_F$ of quantifier depth $1+m+\ell{}q$ that is true in a homogeneous state $\state$ if and only if, for every $i\le L$, $\updateState{\alpha}{\state}$ has type function $F(i)$ if $\alpha$ colors some uncolored element with $c_i$.
  \end{enumerate}
\end{lemma}
\begin{proof}
  \begin{proofenum}
    \item The formula is of the form 
\[
\exists x \bigvee_{i=1}^L (\varphi^i_1\lor \varphi^i_2),
\]
where $\varphi^i_1$ expresses that the state that results from coloring an uncolored element by $c_i$ is not homogeneous and $\varphi^i_2$ expresses that it is not \lhi. 

To this end, $\varphi^i_1$ existentially quantifies two $m$-tuples  (depth: $2m$) and expresses that they have the same $\inpSchema$-type but different $\auxSchema$-types in the state after the coloring (depth: $\ell{}q$). 

The formula $\varphi^i_2$ is a three-fold disjunction for the conditions (H1-3). As an example, the formula for (H1) quantifies two elements $a,a'$ (depth: 2), an $m$-tuple (depth: $m$) and tests that for some color $c_i$ the $\auxSchema$-types of the two databases resulting from the two possible orders in which $a$ and $a'$ can be colored by $c_i$ (depth: $2q$) differ in the $m$-tuple.

Altogether, $\phibad$ has quantifier-depth $1+\max(2m+\ell{}q,2+m+2q)\le 3+2m+(\ell+1)q$.

    \item Similarly, each formula $\varphi_F$ existentially quantifies an element $a$ to be colored, has a disjunct for all possible colors, and universally quantifies an $m$-tuple and tests that the $\auxSchema$-type of it is consistent with its $\inpSchema$-type and $F$.   Overall this yields quantifier depth $1+m+\ell{}q$.
\end{proofenum}%
\end{proof}

We next formalize the observation that for a homogeneous state, the truth of first-order formulas of quantifier depth $k$ only depends on its color frequencies up to $k$\footnote{Note the similarities to Lemma \ref{lem:PropCons11}}.  To this end, we associate with every state $\state$ its \emph{\invectorterm} $\invector{\state} = (n_0, \ldots, n_L)$ over $\N$ where $n_i\df \invectorcomp{\state}{i}$ is the number of elements with $\inpSchema$-color $c_i$ in $\state$. 

We write $n\simeq_k m$, for numbers $k,n,m$, if $n=m$ or both $n\ge k$ and $m\ge k$. We write $(n_0,\ldots,n_L)\simeq_k (n'_0,\ldots,n'_L)$, if for every $i\le L$, $n_i\simeq_k n'_i$.

For a given $k$, we say that two homogeneous states $\state$ and $\state'$ are \emph{$k$-similar} (denoted by $\state\sim_k\state'$) if
\begin{itemize}\item $\invector{\state}\simeq_k \invector{\state'}$ and
\item $\state$ and $\state'$ have the same type function. 
\end{itemize}

Now we can make the relationship between \invectortermplural{} and first-order types more precise.\footnote{We note that for homogeneous states it actually holds:  $\state \sim_k \state'$ if and only if $\state\equiv_k \state'$.}

\begin{lemma}\label{lem:simequiv}
Let $\prog$ be a $\DynFOIA{1}{m}$-program and let $\state$ and $\state'$ be two homogeneous states for $\prog$. For every $k \in \N$, if $\state \sim_k \state'$ then $\state\equiv_k \state'$.
\end{lemma}
We recall that $\state\equiv_k \state'$ means that the two states satisfy exactly the same first-order formulas of quantifier depth (up to) $k$.
\begin{proof}
If $\state \sim_k \state'$ then the duplicator has a straightforward winning strategy for the $k$-round Ehrenfeucht-Fra\"isse game on the $\inpSchema$-reducts of $\state$ and $\state'$. Since both states are homogeneous and have the same type function, this winning strategy extends to $\auxSchema$  and the strategy of duplicator is a winning strategy for $\state$ and $\state'$.
\end{proof}

By combining Lemmas \ref{lem:formulas} and \ref{lem:simequiv} we get the following lemma, which will be the most important technical tool in the proof of a small counterexample property for programs that are not history independent.

\begin{lemma}\label{lem:typefunction}
   Let $\prog$ be a $\DynFOIA{1}{m}$-program with maximum quantifier-depth $q$ and $\ell$ input relations, let $K\ge 1+m+\ell{}q$ and let 
$\state$ and $\state'$ be two homogeneous states for $\prog$ with $\state \sim_K \state'$. Let $a$ and $a'$ be uncolored elements in $\state$ and $\state'$, respectively. Let $\beta$ and $\beta'$ be insertion sequences that color $a$ and $a'$, respectively with the same color $c_i$. Then $\updateState{\beta}{\state}$ and  $\updateState{\beta'}{\state'}$ have the same type function, in case they are both homogeneous. 
\end{lemma}
\begin{proof}
  By Lemma \ref{lem:simequiv}, we know that $\state\equiv_K \state'$. In particular, thanks to Lemma \ref{lem:formulas} and the homogeneity of $\updateState{\beta}{\state}$ and  $\updateState{\beta'}{\state'}$, there is a unique type forecast function $F$ such that $\varphi_F$ holds in \state and $\state'$. Therefore, after coloring $a$ and $a'$ with $c_i$ the resulting states both have type function $F(i)$. 
\end{proof}

Now we can show a small counterexample property for programs that are not history independent. %

\begin{proposition}\label{prop:smallmodel-app}
   Let \prog be a $\DynFOIA{1}{m}$-program with quantifier depth $q$ and $\ell$ input relations, and let $K\df 3+2m+(\ell+1)q$ and $T$ be the number of type functions. If \prog is \emph{not} history independent, then there exists a database $\emptyDB$ of size at most $N\df (2K+T)(L+1)$  and a insertion sequence in normal form $\alpha$ such that $\updateStateI{\alpha}{\emptyDB}$ is bad.
\end{proposition}
\begin{proof}
Let \prog be a dynamic $\DynFOIA{1}{m}$-program that is not history independent and let $\emptyDB$ be an empty database of minimal size $n$ for which there exists a insertion sequence in normal form $\alpha_1\cdots\alpha_N$, such that $\updateState{\alpha}{\emptyDB}$ is bad, each subsequence  $\alpha_i$ colors one element, and $N$ is minimal.

We consider the state $\state\df \updateStateI{\alpha_1\cdots\alpha_{N-1}}{\emptyDB}$ just before the bad state. Thus $\state$ satisfies the formula $\phibad$ from Lemma \ref{lem:formulas}.

Let $ (n_0, \ldots, n_L)\df \invector{\state}$.
We show first that, for every $i\ge 1$,   $n_i\le 2K+T$.
Towards a contradiction, let us assume that for some $i\ge 1$, $n_i>2K+T$. 

Let $\alpha'=\beta \alpha'_1\cdots \alpha'_{n_i}$ be a reordering of $\alpha_1\cdots\alpha_{N-1}$ such that $\alpha'_1,\ldots,\alpha'_{n_i}$ are insertion subsequences that color the $n_i$ elements with color $c_i$ and $\beta$  contains all other insertions. By minimality of $N$, all involved states are \lhi and therefore the reordering does not affect the resulting state, i.e.,  $\updateStateI{\beta\alpha'_1\cdots\alpha'_{n_i}}{\emptyDB}=\state$.

We denote, for every $j\le n_i$, the state $\updateStateI{\beta\alpha'_1\cdots\alpha'_{j}}{\emptyDB}$ by $\state_j$ and its type function by $f_j$. 
We can conclude that $\state_{j}\simeq_K \state_{j'}$, for all $K\leq j<j'\leq n_i-K-1$, since
\begin{itemize}\item 
in $\state_K$, there are more than $K+T$ uncolored elements and $K$
  elements of color $c_i$, 
\item $\alpha'_{K+1}\cdots\alpha'_{n_i-K-1}$ only
  colors uncolored elements with color $c_i$, and
\item  in $\state_{n_i-K-1}$  there are still more then $K$ uncolored elements.
\end{itemize}

Since there are more than $T$ states between $\state_{K}$ and $\state_{n_i-K-1}$, two of them must have the same type function. That is,  there must be $j_1,j_2$ with $K\leq j_1<j_2 \leq n_i-K-1$ and $f_{j_1}=f_{j_2}$ and therefore $\state_{j_1}\sim_K\state_{j_2}$. 

Let $\emptyDB'$ be the empty database  resulting from $\emptyDB$ by deleting all elements that are colored by the sequence  $\alpha'_{j_1+1}\cdots\alpha'_{j_2}$. Since  $\emptyDB'$ has more than $j_1+K>K>q$ elements,  $\state_\init(\emptyDB') \sim_K \state_\init(\emptyDB)$, in particular these two states have the same type functions. By inductive application of Lemma \ref{lem:typefunction} it is easy to show that 
$\updateState{\beta\alpha'_1\cdots\alpha'_{j_1}}{\emptyDB'}\sim_K \updateState{\beta\alpha'_1\cdots\alpha'_{j_1}}{\emptyDB}$.

 In the inductive step, we start from two corresponding states whose $\sim_K$-equivalence has already been established. In particular, they  agree on all formulas $\varphi_F$ and therefore the application of the same one element coloring sequence yields for both the same type function, thanks to Lemma \ref{lem:typefunction} and because the reached states are homogeneous by minimality of $n$ and $N$.
Since the number of elements for each (proper) color is the same in both new states and both have more than $K$ uncolored elements, they are also equivalent with respect to $\simeq_K$.   

For each $j$ with $j_2\le j\le N-1$ let $\state'_j\df \updateState{\beta\alpha'_1\cdots\alpha'_{j_1}\alpha'_{j_2+1}\cdots\alpha'_{j}}{\emptyDB'}$. 

We emphasize that, for every $j$, $\invector{\state'_j}$ and $\invector{\state_j}$ only differ in their entry for color $c_i$ (which for both is at least $K$). In particular, they have the same number of uncolored elements.

Thus, $\state'_{j_2}\sim_K \state_{j_1}\sim_K \state_{j_2}$ and therefore, as before, $\state'_{j_2}$ and $\state_{j_2}$ agree on all formulas $\varphi_F$.  It follows that the two states $\state'_{j_2+1}$ and $\state_{j_2+1}$ obtained by the sequence $\alpha'_{j_2+1}$  again have the same type function.  As they both have at least $K$ uncolored elements and at least $K$ elements with color $c_i$ (and agree on all other color frequencies), we get $\state'_{j_2+1}\sim_K\state_{j_2+1}$. An inductive application of the same argument yields $\state'_{N-1}\sim_K\state_{N-1}=\state$. Since $\state\models\phibad$ we conclude $\state'_{N-1}\models\phibad$ and thus  $\state'_{N-1}$ is a bad state. As  $\state'_{N-1}$ can be reached by fewer insertions than $\state$ we get the desired contradiction and thus $n_i\le 2K+T$, for all $i\ge 1$.\\

We finally show that $n_0 \le K$.  Otherwise, if $n_0>K$, we could replace $\emptyDB$ by the empty database $\emptyDB'$ in which one element that is uncolored in $\state$ is removed. 
Similarly as before it would follow that $\updateState{\alpha_1\cdots\alpha_{N-1}}{\emptyDB'}\sim_K \updateState{\alpha_1\cdots\alpha_{N-1}}{\emptyDB}$ and therefore that $\updateState{\alpha_1\cdots\alpha_{N-1}}{\emptyDB'}$ satisfies $\phibad$ and is therefore bad, contradicting the choice of $\emptyDB$. This completes the proof of the proposition. 
\end{proof}

We can now conclude the main result of this section.

\begin{theorem}\label{theorem:hidecidable}
\HI is decidable for $\DynFOI{1}$-programs.
\end{theorem}

\begin{proof}
It follows immediately from Proposition \ref{prop:smallmodel-app} that Algorithm \ref{algorithm:hiunary} is a correct decision algorithm for \HI of $\DynFOI{1}$-programs.

\begin{algorithm}
    \caption{Deciding \HI for $\DynFOI{1}$-programs}\label{algorithm:hiunary}
    \begin{algorithmic}[1]
        \INPUT A $\DynFOIA{1}{m}$-program $\prog$ with $\ell$ input relations and quantifier depth $q$.
        \State Let $K$, $L$ and $T$ be as in Proposition \ref{prop:smallmodel-app}.
       \For{all empty databases $\emptyDB$ over domains $\{1,\ldots,n\}$ with $n\le (2K+T)(L+1)$}
        \For{all normal form insertion sequences $\alpha$ over $\{1,\ldots,n\}$}
          \LineIf {$\updateStateI{\alpha}{\emptyDB}$ is not homogeneous or not \lhi} {Reject.}
        \EndFor
       \EndFor
        \State Accept.
    \end{algorithmic}
  \end{algorithm}
\end{proof}

Using the same technique as used in the proof of Theorem \ref{theorem:emptiness:consistent:unaryDynProp}(b), history independence can be shown to be decidable for $\DynPropA{1}$-programs.
\begin{theorem}\label{theorem:hidynpropdecidable}
  \HI is decidable for $\DynPropA{1}$-programs.  
\end{theorem}
\begin{proof}
  Let $\prog$ be a $\DynPropIA{\ell}{1}$-program for some $\ell \in \N$. Recall that, according to Lemma \ref{lemma:hicharacterization}, for testing history independence it suffices to check that no non-locally history independent state can be reached by an insertion sequence in normal form. We argue that if a non-locally history independent state can be reached by $\prog$, then such a state with few tuples in the input relations can be reached as well. History independence can then be tested in a brute force manner by trying out insertion sequences for all input databases with few tuples.
  
  Suppose that $\state$ is a non-locally history independent state reachable by $\prog$ such that the number $N$ of tuples in input databases of $\state$ is minimal. In particular, $\prog$ is history independent for input databases with less than $N$ tuples, that is, all modification sequences $\alpha$ and $\alpha'$ yielding an input database with less than $N$ tuples also yield the same state. Let $\vec a$ be an $2\ell$-ary tuple that witnesses that $\state$ is not locally history independent, i.e.~there are two modifications on $\vec a$ that contradict (H1), (H2) or (H3). Further let $\gamma$ be the atomic type of $\vec a$. Now, using the same argument as in the proof of Theorem \ref{theorem:emptiness:consistent:unaryDynProp} as well as the history independence of $\prog$ for databases with less than $N$ tuples, one can show that for exhibiting a tuple of type $\gamma$ the number $N$ of input tuples does not have to be large. %
\end{proof}
}

\section{Conclusion}\label{section:conclusion}
\makeatletter{}%
In this work we studied the algorithmic properties of static analysis problems for (restrictions of) dynamic programs. Most of the results are summarized in Table~\ref{tab:results}. %
In general only very strong restrictions yield decidability.

The only cases left open are about $\DynProp$-programs when both the arity of the input and the arity of the auxiliary relations is at least~2. For such programs the status of history independence and emptiness of consistent remains open. We conjecture that for history independence the decidable fragment of $\DynProp$ is larger than exhibited here.

Our results will hopefully contribute to a better understanding of the power of dynamic programs. On the one hand the undecidability proofs show that very restricted dynamic programs can already simulate powerful machine models. It is natural to ask whether this power can be used to maintain other, more common queries. On the other hand the decidability results utilize limitations of the state space and the transition between states for classes of restricted programs. Such limitations can be a good starting point for the development of techniques for proving lower bounds for the respective fragments.

 \bibliography{bibliography}

 \shortVersion{
  }
\end{document}